%% file: maximinpara.arxiv-v1.tex
\documentclass{scrartcl}

\PassOptionsToPackage{nameinlink}{cleveref}

\usepackage[todo,firstnames]{group}

\usepackage[margin=2.5cm]{geometry}

\usepackage{mathtools}
\usepackage{theoremref}
\usepackage{pgfplots}
\usepackage{cancel}
\usepackage{xcolor}
\usepackage{graphicx} %

\hypersetup{
	pdfencoding=auto, 
	psdextra,
	colorlinks=true,
	citecolor=green!50!black,
	linkcolor=red!60!black,
	urlcolor=blue!90!black
}

\newtheorem{fact}{Fact}

\usepackage{subcaption}
\usepackage{makecell} %

\newcommand{\condorcet}[1]{\textcolor{red!75!black}{#1}}

\title{Condorcet-Consistent Choice\\ Among Three Candidates}

\usepackage{authblk}

\author[1]{Felix Brandt}
\author[1]{Chris Dong}
\author[2]{Dominik Peters}
\affil[1]{Technische Universität München} 
\affil[2]{CNRS, LAMSADE, Universit\'e Paris Dauphine - PSL}
\date{\vspace{-1cm}}

\begin{document}

\maketitle

\begin{abstract}	
A voting rule is a Condorcet extension if it returns a candidate that beats every other candidate in pairwise majority comparisons whenever one exists. Condorcet extensions have faced criticism due to their susceptibility to variable-electorate paradoxes, especially the reinforcement paradox \citep{YoLe78a} and the no-show paradox \citep{Moul88b}. In this paper, we investigate the susceptibility of Condorcet extensions to these paradoxes for the case of exactly three candidates.
For the reinforcement paradox, we establish that it must occur for every Condorcet extension when there are \emph{at least eight} voters and demonstrate that certain refinements of maximin---a voting rule originally proposed by \citet{Cond85a}---are immune to this paradox when there \emph{are at most seven} voters. For the no-show paradox, we prove that the \emph{only} homogeneous Condorcet extensions immune to it are refinements of maximin. We also provide axiomatic characterizations of maximin and two of its refinements, Nanson's rule and leximin, highlighting their suitability for three-candidate elections.
\end{abstract}

\section{Introduction}

Deciding between two candidates based on the preferences of multiple voters is straightforward and allows for simple and natural rules that satisfy virtually all desirable properties, such as majority rule or weighted threshold rules.
However, \emph{choosing from three or more candidates} is what leads to significant challenges and inevitable tradeoffs. Arrow's impossibility and the Gibbard-Satterthwaite theorem are the best known among many results highlighting these difficulties \citep{Arro51a,Gibb73a,Satt75a}. These findings help to explain the 
plethora
of voting rules proposed for elections with three or more candidates. In this paper, we address whether identifying a suitable rule becomes easier when focussing on the case of \emph{exactly three candidates}.

The primary reason why plurality rule---the most common voting rule---struggles with three or more candidates is that it only takes into account the voters' top choices and thereby, for example, ignores what a voter who favors candidate $a$ thinks about candidates $b$ and $c$. Once this information is taken into account, \citet{Cond85a} argued that the best choice is the candidate that the majority prefers, in the sense that every other candidate is judged to be worse by some majority of voters. Nowadays, such a candidate is called a \emph{Condorcet winner}. Choosing a Condorcet winner is motivated by the observation that any other candidate can be overthrown by a coordinated majority of voters who all prefer the same candidate to the selected candidate. However, Condorcet recognized that a Condorcet winner need not exist since the majority relation may be cyclic. This has given rise to extensive efforts designing \emph{Condorcet extensions}, voting rules that select the Condorcet winner whenever it exists and use some other procedure to choose a candidate in the remaining cases. Dozens of Condorcet extensions have been cataloged \citep{Fish77a,Lasl97a,BCE+14a}.
Condorcet also proposed a Condorcet extension of his own, clearly described for the case of three candidates: 

\begin{quote}
``When the three [pairwise majority] views cannot exist together [because of a cycle], the adopted view results from the two [pairwise majority views] that are most probable [i.e., have the largest majority].'' \hfill \mbox{\citep[p.~125]{Cond85a}}
\end{quote}
How this idea should be extended to four or more candidates is only outlined in vague terms and ultimately remains elusive. However, his proposal for three candidates is unambiguous and coincides with what many modern Condorcet extensions do in this case, including maximin, ranked pairs, beat path, split cycle, Kemeny's rule, Dodgson's rule, and Young's rule.

While compelling, Condorcet's principle to always return the Condorcet winner 
has also met some resistance. The two main lines of attack concern settings with a variable set of voters. In particular, 
\citet{YoLe78a} showed that no Condorcet extension satisfies \emph{reinforcement}, a consistency condition demanding that if the same winner is returned for two different electorates, it should also be returned for the union of these electorates.
In a similar spirit, \citet{Moul88b} demonstrated that every Condorcet extension suffers from the \emph{no-show paradox}, where a voter might benefit by abstaining from an election to achieve a more desirable outcome. 
The proofs of \citeauthor{YoLe78a} and \citeauthor{Moul88b} are based on constructing collections of preference profiles on which the paradoxes cannot be avoided. \citeauthor{Moul88b}'s proof uses profiles with at least four candidates and 25 voters. Later, a smaller proof requiring only 12 voters was found, and no proof with fewer voters exists \citep{BGP16c}. \citeauthor{YoLe78a}'s proof for the reinforcement paradox uses at least three candidates and 13 voters.

In this paper, we show that refinements of maximin---the rule proposed by Condorcet---perform particularly well with respect to these paradoxes when there are exactly three candidates. 
Specifically, for up to seven voters, we exhibit a refinement of maximin that is immune to the reinforcement paradox.
During a June 2022 lecture at the Institut Henri Poincaré (``Mathématiques et Démocratie'' conference), Hervé Moulin \href{https://sciencesmaths-paris.fr/images/pdf/MEM2022SLIDES/Herve.pdf\#page=10}{asked} what the smallest number of voters is for which \citeauthor{YoLe78a}'s result holds.
We use SAT solvers to answer this question, showing that it holds even for eight voters. This bound is tight, and thus certain maximin refinements avoid the reinforcement paradox as much as possible. 
Regarding the no-show paradox, it is known that maximin with a fixed tie-breaking order avoids the paradox for three candidates \citep{Moul88b}. We prove that the \emph{only} homogeneous Condorcet extensions immune to the no-show paradox are refinements of maximin. Moreover, by adding continuity, we obtain an axiomatic characterization of maximin. We also characterize two refinements of maximin (Nanson's rule and leximin) using natural monotonicity and invariance axioms.

The rest of the paper is structured as follows.
In \Cref{sec:prelims}, we introduce our model and standard axioms. 
In \Cref{sec:condorcet-extensions}, we provide an overview of Condorcet extensions and their inclusion relationships when only three candidates exist. 
In \Cref{sec:reinforcement}, we present several impossibility results involving reinforcement, each using the minimum possible number of voters. 
In \Cref{sec:participation}, we show that maximin refinements are the only homogeneous Condorcet extensions that circumvent the no-show paradox and provide characterizations of maximin, Nanson's rule, and leximin.

\section{Preliminaries}
\label{sec:prelims}

Let $A = \{a, b, c \}$ be the set of three \emph{candidates} and let $N^*$ be the (possibly infinite) set of potential \emph{voters}.
A non-empty and finite set $N\subseteq N^*$ is called an \emph{electorate}. Given an electorate $N$, a \emph{(preference) profile} $P=(\succ_i)_{i\in N}$ over $N$ is a collection of linear orders, with $\succ_i$ denoting the preference of voter $i \in N$.
For example, if $a \succ_i b \succ_i c$, then voter $i$ prefers $a$ the most, followed by $b$, with $c$ preferred least. We sometimes write $abc$ as a shorthand for this linear order; in \Cref{fig:prelim-example-profile} and in other figures we show linear orders as columns.
We will write $x \succeq_i y$ if $x = y$ or $x \succ_i y$.

Let $P$ be a profile. For two candidates $x, y \in A$, we define the \emph{margin} of $x$ over $y$ as 
\[
m_{x,y}(P) = |\{i \in N : x \succ_i y\}| - |\{i \in N : y \succ_i x\}|,
\]
and we just write $m_{x,y}$ if $P$ is clear from the context.
Thus, if $m_{x,y} > 0$, then a strict majority of voters prefer $x$ to $y$. Note that $m_{x,y} = -m_{y,x}$ and that $m_{x,x} = 0$. We can display this information in a \emph{margin graph}, which is a weighted digraph on vertex set $A$ where for each pair $\{x, y\}$ of vertices, we orient the edge as $x \to y$ in such a way that $m_{x, y} > 0$ (so that $x$ beats $y$ in a pairwise comparison) and give it weight $m_{x,y}$. For pairs with $m_{x, y} = 0$, we do not draw an edge. See \Cref{fig:prelim-example-margin-graph} for an example of a margin graph. For our purposes, it is often not necessary to know the precise weight of an edge, and in those cases, we can draw the \emph{ordinal margin graph}, where we preserve the direction of each edge and show which edges have strictly higher weights than others (by drawing the higher-weight edge with more thickness) and which edges have equal weight (by drawing the edges with the same thickness). Sometimes we use an empty arrowhead for the thinnest edges. See \Cref{fig:prelim-example-ordinal-margin-graph} for an example of an ordinal margin graph. 

For any profile $P$ over $N$, the margins $(m_{x,y})_{x,y \in A}$ all have the same parity (which is the same as the parity of $|N|$). The well-known \emph{McGarvey theorem} states that for every margin graph with same-parity weights, there is a profile $P$ that induces it \citep{McGa53a,Debo87a}.

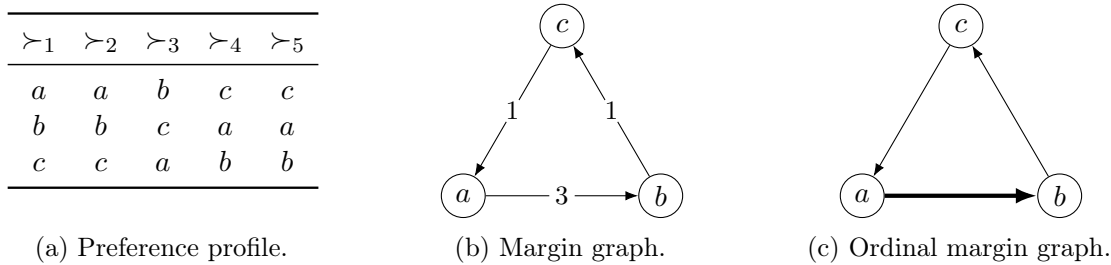
\begin{figure}[t]
	\centering
	\begin{subfigure}{0.28\linewidth}
		\centering
		\[
		\begin{array}{ccccc}
			\toprule
			{\succ_1} & {\succ_2} & {\succ_3} & {\succ_4} & {\succ_5}\\
			\midrule
			a & a & b & c & c \\
			b & b & c & a & a \\
			c & c & a & b & b \\
			\bottomrule
		\end{array}
		\]
		\caption{Preference profile.}
		\label{fig:prelim-example-profile}
	\end{subfigure}\qquad
	\begin{subfigure}{0.28\linewidth}
		\centering
		\begin{tikzpicture}[baseline=-1ex]
			\tikzstyle{mynode}=[fill=white,circle,draw,minimum size=1.5em,inner sep=2pt]
			\tikzstyle{mylabel}=[fill=white,circle,inner sep=0.5pt]
			\node[mynode] (a) at (210:1.5) {$a$};
			\node[mynode] (c) at (90:1.5) {$c$};
			\node[mynode] (b) at (330: 1.5) {$b$};
			
			\draw[-Latex] (a) edge node[mylabel] {\small$3$} (b);
			\draw[-Latex] (b) edge node[mylabel] {\small$1$} (c);
			\draw[-Latex] (c) edge node[mylabel] {\small$1$} (a);
		\end{tikzpicture}
		\caption{Margin graph.}
		\label{fig:prelim-example-margin-graph}
	\end{subfigure}\qquad
	\begin{subfigure}{0.28\linewidth}
		\centering
		\begin{tikzpicture}[baseline=-1ex]
			\tikzstyle{mynode}=[fill=white,circle,draw,minimum size=1.5em,inner sep=0pt]
			\tikzstyle{mylabel}=[fill=white,circle,inner sep=0.5pt]
			\node[mynode] (a) at (210:1.5) {$a$};
			\node[mynode] (c) at (90:1.5) {$c$};
			\node[mynode] (b) at (330: 1.5) {$b$};
			
			\draw[-{Latex[length=8pt]}] (a) edge [ultra thick] (b);
			\draw[-Latex] (b) edge (c);
			\draw[-Latex] (c) edge (a);
		\end{tikzpicture}
		\caption{Ordinal margin graph.}
		\label{fig:prelim-example-ordinal-margin-graph}
	\end{subfigure}
	\caption{An example of a preference profile over $N = \{1, \dots, 5\}$ and its margin graph and ordinal margin graph. The ordinal margin graph encodes the information that $m_{a,b} > m_{b,c} = m_{c, a} > 0$.}
	\label{fig:prelim-example}
\end{figure}

A \emph{social choice function} $f$ maps each profile $P$ (on any electorate $N\subseteq N^*$) to a non-empty set $f(P) \subseteq A$ of winners.
It is \emph{resolute} if $\lvert f(P)\rvert=1$ for all profiles $P$. 
A social choice function $f_1$ \emph{refines} another social choice function $f_2$ if $f_1(P) \subseteq f_2(P)$ for all profiles $P$.

Let us define four basic properties of a social choice function. The first two are standard symmetry conditions.
\begin{itemize}
	\item A social choice function $f$ is \emph{anonymous} if its output does not depend on the identity of the voters. Formally, for all profiles $P$ over $N$ and all injections $\pi\colon  N \to N^*$, we have $f(P) = f(\pi(P))$ where $\pi(P)$ is the profile over $\pi(N)$ where voter $\pi(i)$ has preferences $\succ_{i}$.
	\item A social choice function $f$ is \emph{neutral} if its output does not depend on the identities of the candidates. Formally, for all profiles $P$ and all bijections $\sigma\colon A \to A$, we have $\sigma(f(P)) = f(\sigma(P))$, where $\sigma(P)$ refers to the profile obtained from $P$ by replacing each voter $i$'s preference $x \succ_i y \succ_i z$ with $\sigma(x) \succ_i \sigma(y) \succ_i \sigma(z)$. 
\end{itemize}
The next two properties involve the operation of copying a profile, which is always possible, provided that the set of potential voters is infinite. Thus, we will only invoke these axioms when $N^* = \mathbb N$.
Let $P$ be a profile over $N$ with $|N| = n$. We say that a profile $Q$ over electorate $N'$ with $|N'| = t \cdot n$ is a \emph{$t$-fold copy} of $P$ if we can partition $N'$ into $t$ disjoint subelectorates $N_1, \dots, N_t$, each of size $n$, such that the restricted profile $Q_{N_j}$ is a copy of $P$ for each $j = 1, \dots, t$, in the sense that there is a bijection $\phi : N \to N_j$ such that each voter $i \in N$ has the same preferences as $\phi(i)$. We usually write ``$tP$'' to denote a profile that is a $t$-fold copy of $P$.
\begin{itemize}
	\item A social choice function $f$ is \emph{homogeneous} if doubling the profile does not affect the output. Formally, $f$ is homogeneous if for all profiles $P$ and $2P$ such that $2P$ is a 2-fold copy of $P$, we have $f(P)= f(2P)$.
	\item A social choice function $f$ is \emph{continuous} if for all  profiles $P$ and $P'$, there is some $n'\in \mathbb N$ such that for all $n\ge n'$ it holds that $f(nP + P')\subseteq f(P)$ whenever $nP$ is an $n$-fold copy of $P$ defined on an electorate disjoint from $P'$.
\end{itemize}
\citet{Youn75a} introduced the continuity axiom for anonymous social choice functions, with the intuition that overwhelmingly large groups of voters should decide the outcome, while negligible, arbitrarily small groups should only be able to break ties. It is also known as the ``overwhelming majorities axiom'' \citep{Myer95b}.

\section{Condorcet Extensions}
\label{sec:condorcet-extensions}
We say that a candidate $x \in A$ is a \emph{Condorcet winner} if $m_{x,y} > 0$ for all $y \in A \setminus \{x\}$, which means that $x$ beats every other candidate in a pairwise majority comparison. A Condorcet winner may fail to exist (see \Cref{fig:prelim-example}), but if it exists, it has to be unique.
A social choice function $f$ is a \emph{Condorcet extension} or \emph{Condorcet-consistent} if $f(P) = \{x\}$ whenever $x$ is a Condorcet winner.

Numerous Condorcet extensions have been defined, though for the case of three candidates, many of these rules coincide. Of particular interest is the \emph{maximin rule}, which selects the candidates whose worst margin is highest:
\[
f_{\text{maximin}}(P) = \arg\max_{x \in A} \min_{y \in A \setminus \{x\}} m_{x, y}.
\]
Thus, the winners under this rule never lose too badly against another candidate. It is sometimes also known as the \emph{minimax} rule (since it minimizes the worst loss) or the \emph{Simpson-Kramer} rule. 

To define some other social choice functions, we need the concept of the \emph{Borda score} of a candidate, defined as $\beta_x = \sum_{y \in A} m_{x,y}$. This is the net out-degree of $x$ in the margin graph. The \emph{Borda rule} selects the candidates with the highest Borda score. Note that a Condorcet winner $x$, if one exists, has a strictly positive Borda score (because $m_{x,y} > 0$ for all $y \neq x$). In addition, because $m_{x,y} = -m_{y,x}$, we have $\sum_{x \in A} \beta_x = 0$, so the average Borda score of a candidate is 0, and thus a Condorcet winner always has an above-average Borda score. However, it need not have the \emph{highest} Borda score. Still, Borda scores can be used to define some interesting Condorcet extensions:

\begin{itemize}
	\item The \citet{Blac48a} rule is the rule that selects the Condorcet winner if one exists and otherwise returns the candidates with the highest Borda score.
	\item The \citet{Nans82a} rule repeatedly deletes all candidates whose Borda score is not positive. Borda scores are computed with respect to the remaining candidates $A'\subseteq A$, and the process continues until there is no candidate $x$ with positive Borda score $\sum_{y \in A'} m_{x,y}$.
	\item The \emph{leximin rule} selects candidates whose worst margin is highest, just like maximin. However, when there are several such candidates, it breaks the tie in favor of those with the higher second-worst margin, breaking any remaining ties using the third-worst margin, and so on.%
	\footnote{\label{fn:leximin definition}Formally, for $x \in A$, let $m^x = (m^x_{1},\dots, m^x_{m})$ be a reordered vector of the majority margins $(m_{x,y})_{y\in A}$ with $m^x_{1}\le \dots\le m^x_{m}$. For $x,y\in A$ we write $m^x \triangleright m^y$ if and only if there is some $r\le m$ such that $m^x_{i} = m^y_{i}$ for all $i< r$ and $m^x_{r}>m^y_{r}$. For a given profile, leximin returns the set of all $x$ such that there is no $y\in A$ with $m^y \triangleright m^x$. We call $\triangleright$ the \emph{lexicographic ordering}. It forms a linear order over $\mathbb Z^m$, so for leximin to return multiple candidates, these must have the same ordered vector of majority margins.}
	This is a natural way to refine maximin, though it does not appear to have been studied in the academic literature.%
	\footnote{There have been discussions of this rule on the \texttt{election-methods} mailing list (\href{http://lists.electorama.com/pipermail/election-methods-electorama.com//2010-June/091958.html}{2010}, \href{http://lists.electorama.com/pipermail/election-methods-electorama.com/2011-November/094364.html}{2011}), and a preprint finds that leximin is more frequently resolute than other maximin refinements in simulations \citep{Darl16a}.}
	In the case of three candidates, the leximin rule is equivalent to the maximin rule with ties broken in favor of candidates with higher Borda score, a rule that was also discussed by \citet[Appendix A]{HoPa23a}. To see this, note that if two candidates $x$ and $y$ have the same worst margin, then $x$ has a higher second-worst margin than $y$ if and only if the sum of these two margins is higher for $x$ than for $y$.
\end{itemize}

\begin{figure}[t]
	\centering
	\includegraphics[width=11cm]{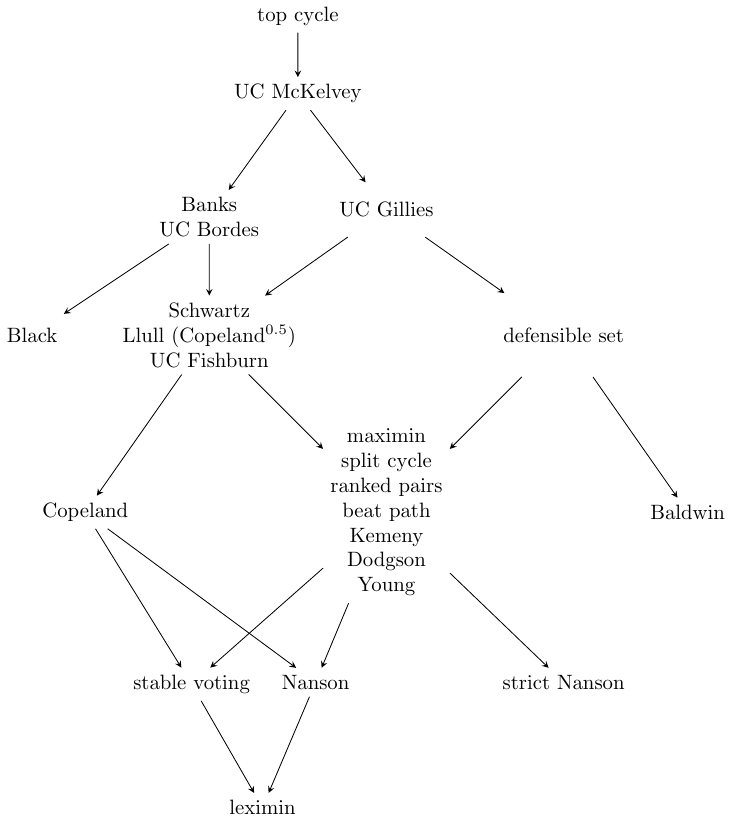}
	\caption{Hasse diagram of three-candidate Condorcet extensions where lower rules refine higher ones, and rules at the same node are identical.}
	\label{fig:hasse}
\end{figure}

There are many other commonly studied Condorcet extensions. For brevity, we will not include definitions of all the rules we mention below since they can be found elsewhere (\citealp{Dugg11a}, for variants of the uncovered set (UC); \citealp{HoPa21b}, for top cycle, Llull, Copeland, Baldwin, strict Nanson, ranked pairs, beat path, split cycle; \citealp{CHH15a}, for Dodgson and Young; \citealp{HoPa22a}, for stable voting; \citealp{Holl24a} or \Cref{sec:participation} below, for the defensible set).

We will now study how all these Condorcet extensions relate to each other in the case of three candidates. In particular, we will be interested in which rules are equivalent (i.e., select the same output for all profiles) and which rules refine which other rules.
\Cref{fig:hasse} shows a Hasse diagram of the refinement relation of these Condorcet extensions. Rules that appear in the same node of the diagram are equivalent. Note, in particular, the large cluster of equivalent Condorcet extensions in the middle: maximin is equivalent to split cycle, ranked pairs, beat path, as well as the rules of Kemeny, Dodgson, and Young. 

\begin{table}
	\makebox[\linewidth][c]{
		\scalebox{0.7}{
			\input{rule-table-ordinal-margins.tex}
	}}
	\caption{The outputs of rules that only depend on the ordinal margin graph.}
	\label{tbl:rules-ordinal}
\end{table}

\begin{theorem}
	\label{thm:hasse}
	For three candidates, \Cref{fig:hasse} shows the Hasse diagram of the refinement relation between select Condorcet extensions.
\end{theorem}
\begin{proof}
	Except for Black's and Baldwin's rules (which we will consider separately), all of the social choice functions included in the diagram depend only on the ordinal margin graph. This means that if two profiles induce the same ordinal margin graph, then the output is the same for the two profiles \citep{Holl24a}. For most rules, this can be seen straightforwardly from their definitions. In fact, many of the Condorcet extensions in the diagram only depend on the orientation of the edges in the margin graph and not their relative weights. For Kemeny's, Dodgson's, and Young's rules, this follows because they are equivalent to maximin for three candidates (\citealp{CMM14a}, working paper version, Theorem 1; see also \citealp{Heil20a}, Theorem 3.2). For the Nanson and strict Nanson rules, this follows because, for three candidates, the ordinal margin graph contains enough information to determine whether a candidate has below-average Borda score.%
	\footnote{Nanson's rule eliminates all candidates whose Borda score is at most 0 (or strictly below 0 for the strict Nanson rule). For three candidates, a candidate's Borda score is at most 0 if and only if it either does not have outgoing edges or it has one outgoing edge that is lighter than its incoming edge. It is strictly below 0 if it either has no outgoing edges and an incoming edge or it has one outgoing edge that is strictly lighter than its incoming edge.}
	Note that all these rules are Condorcet extensions, and they are neutral. Therefore, these rules can be fully specified by listing their outputs on the 12 possible non-isomorphic ordinal margin graphs in which there is no Condorcet winner. We have produced such a list in \Cref{tbl:rules-ordinal}, which immediately establishes most of the relations displayed in \Cref{fig:hasse}.
	
	For Black's rule, note that it will never select a candidate that only has incoming edges but no outgoing edges (since such a candidate has negative Borda score). Inspecting \Cref{tbl:rules-ordinal}, we see that, when there is no Condorcet winner, Banks' rule returns all candidates for which this is not the case (among others); hence Black's rule refines Banks (which in turn refines UC McKelvey and the top cycle). On the profile \href{https://voting.ml/?profile=3abc-1bca-4cab}{$3abc + bca + 4cab$}, Black's rule selects $\{a\}$ while UC Gillies selects $\{b,c\}$; hence Black's rule is not refined by and does not refine any other rule in the diagram (which are all refinements of UC Gillies).
	
	For Baldwin's rule, we can go through the 12 ordinal margin graphs in \Cref{tbl:rules-ordinal} to determine what it may output on each of them and thereby deduce that it refines the defensible set (and consequentially also UC Gillies, UC McKelvey, and the top cycle).%
	\footnote{For Graphs A and B, the defensible set contains all candidates, so Baldwin refines it. For Graphs C, D, E, and F, $b$ is the Borda loser, and so is eliminated first, so Baldwin refines the defensible set $\{a,c\}$. For Graphs G and H, the Borda loser is either $b$ or $c$; if $b$ is eliminated then $c$ wins; if $c$ is eliminated then $a$ wins, refining the defensible set $\{a,c\}$. For Graphs I and J, $c$ is eliminated, followed by $b$, so $a$ wins, refining the defensible set $\{a,c\}$. For Graphs K and L, $c$ is eliminated, followed by $b$, so $a$ wins, refining the defensible set $\{a\}$.}
	On the profile \href{https://voting.ml/?profile=4acb-5bac-3cab-5cba}{$4acb+5bac+3cab+5cba$}, Baldwin selects $\{a\}$ while Black's rule, leximin, and strict Nanson (the leaves of the diagram) all select $\{c\}$, so Baldwin is not refined by any other rule in the diagram. On the other hand, on the profile \href{https://voting.ml/?profile=1ABC-3BCA-4CAB}{$1abc+3bca+4cab$}, Baldwin's rule selects $\{b, c\}$ while Banks' rule selects $\{a, c\}$, so Baldwin does not refine rules other than the ones indicated in the diagram.
\end{proof}

In the absence of majority ties (i.e., when only Graphs A, C, G, I, and K are considered), even more rules coincide. The top cycle, all variants of the uncovered set, Schwartz, and Copeland are equivalent because they return all candidates when the majority graph is a three-cycle. Moreover, stable voting coincides with maximin,
and Nanson's rule coincides with leximin. When restricting attention to the generic case where, on top of the absence of majority ties, no two margins are equal (called \emph{uniquely weighted} profiles by \citealp{HoPa20a}), maximin is resolute and hence coincides with \emph{all} its refinements. 
In \Cref{tbl:rules-ordinal}, only Graph G and Graph K are generic in this sense, and we can see that the rules mentioned in the table collapse to a hierarchy of only three Condorcet extensions: the top cycle, the defensible set, and maximin.

\section{Reinforcement}
\label{sec:reinforcement}

In this section, we study the reinforcement axiom, which demands that candidates who win for two disjoint electorates should be precisely the winners for the union of these electorates \citep{Youn74a,Youn75a}. Suppose that $P=(\succ_i)_{i\in N}$ and $P'=(\succ_i)_{i\in N'}$ are preference profiles defined over electorates $N$ and $N'$ that are disjoint: $N\cap N'=\emptyset$. We can, therefore, merge these profiles into a single profile $P+P'= (\succ_i)_{i\in N\cup N'}$ defined on the electorate $N \cup N'$. 
\begin{definition}[Reinforcement]
A social choice function $f$ satisfies \emph{reinforcement} if for all profiles $P$ and $P'$ on disjoint electorates, we have
$f(P + P') = f(P) \cap f(P')$ whenever $f(P) \cap f(P') \neq \emptyset$.
\end{definition}
In other words, if there is a candidate $a$ who wins in both $P$ and $P'$, then $a$ also wins in the combined profile $P + P'$, and all other winners in $P + P'$ must also win in both $P$ and $P'$.

\citet[Theorem 2]{YoLe78a} famously showed that reinforcement is in conflict with Condorcet-consistency. Their result is only phrased in terms of \emph{weak} Condorcet winners, but they briefly mention a variant of this result that can be turned into a 13-voter impossibility for Condorcet extensions \citep[see, e.g.,][Theorem 9.2]{Moul88a}.
We will show that the incompatibility remains intact even when there are fewer than 13 voters.

\subsection{Simple Proof for Nine Voters}

We begin by giving a very simple proof of the incompatibility for 9 voters. It can be obtained by essentially optimizing the existing proofs of \citet{YoLe78a} and \citet{Moul88a} with respect to the number of required voters. This proof was found by Keyvan Kardel and also appears in the handbook chapter by \citet[Prop.~2.5]{Zwic15a}.

A remarkable feature of this proof is that it only uses a weak form of reinforcement:
A social choice function satisfies \emph{subset-reinforcement} if for all $P_1$ and $P_2$ defined on disjoint electorates, we have $f(P_1) \cap f(P_2) \subseteq f(P_1 + P_2)$. This property was also discussed by \citet{YoLe78a}.

\begin{theorem}
	\label{thm:reinf-9}
	Every Condorcet extension violates subset-reinforcement when $|N^*| \ge 9$.
\end{theorem}
\begin{proof}
	Let $f$ be a Condorcet extension satisfying subset-reinforcement. Let $N_1$ and $N_2$ be two disjoint electorates consisting of $6$ and $3$ voters, respectively. Without loss of generality, let $N_1=\{1,2,3,4,5,6\}$ and $N_2 = \{7,8,9\}$. Consider the following profile $P_1$ with 6 voters, which is a ``double Condorcet cycle'':
	\begin{align*}
	\begin{array}{cccccc}
		\toprule
		{\succ_1} & {\succ_2} & {\succ_3} & {\succ_4} & {\succ_5} & {\succ_6}\\
		\midrule
		a & a & b & b & c & c \\
		b & b & c & c & a & a \\
		c & c & a & a & b & b \\
		\bottomrule
	\end{array}
	\hspace{2.43cm}
	\qquad&
	\begin{tikzpicture}[scale=0.75,baseline=0.2cm]
		\tikzstyle{mynode}=[fill=white,circle,draw,minimum size=1.5em,inner sep=2pt]
		\tikzstyle{mylabel}=[fill=white,circle,inner sep=0.5pt]
		\node[mynode] (a) at (210:1.5) {$a$};
		\node[mynode] (c) at (90:1.5) {$c$};
		\node[mynode] (b) at (330: 1.5) {$b$};
		\draw[-Latex] (a) edge node[mylabel] {\small$2$} (b);
		\draw[-Latex] (b) edge node[mylabel] {\small$2$} (c);
		\draw[-Latex] (c) edge node[mylabel] {\small$2$} (a);
	\end{tikzpicture}
	\intertext{%
	Because $f(P_1) \neq \emptyset$, we may assume without loss of generality that $a \in f(P_1)$ (otherwise relabel the candidates).
	Now consider the following profile $P_2$ with 3 voters:
	}
	\begin{array}{ccc}
		\toprule
		{\succ_7} & {\succ_8} & {\succ_9}\\
		\midrule
		\condorcet{a} & \condorcet{a} & c \\
		c & c & \condorcet{a} \\
		b & b & b \\
		\bottomrule
	\end{array}
	\qquad&
	\begin{tikzpicture}[scale=0.75,baseline=0.2cm]
		\tikzstyle{mynode}=[fill=white,circle,draw,minimum size=1.5em,inner sep=2pt]
		\tikzstyle{mylabel}=[fill=white,circle,inner sep=0.5pt]
		\node[mynode] (a) at (210:1.5) {$\condorcet{a}$};
		\node[mynode] (c) at (90:1.5) {$c$};
		\node[mynode] (b) at (330: 1.5) {$b$};
		\draw[-Latex] (a) edge node[mylabel] {\small$3$} (b);
		\draw[-Latex] (c) edge node[mylabel] {\small$3$} (b);
		\draw[-Latex] (a) edge node[mylabel] {\small$1$} (c);
	\end{tikzpicture}
	\intertext{%
	In this profile, $a$ is the Condorcet winner, so $f(P_2) = \{a\}$ because $f$ is a Condorcet extension.
	Note that $P_1$ and $P_2$ are defined over disjoint electorates, and that $a \in f(P_1) \cap f(P_2)$.
	Because $f$ satisfies subset-reinforcement, we have $f(P_1) \cap f(P_2) \subseteq f(P_1 + P_2)$, and therefore $a \in f(P_1 + P_2)$.
	However, in the combined profile $P_1 + P_2$, $c$ is the Condorcet winner:
	}
	\begin{array}{ccccccccc}
		\toprule
		{\succ_1} & {\succ_2} & {\succ_3} & {\succ_4} & {\succ_5} & {\succ_6} & {\succ_7} & {\succ_8} & {\succ_9} \\
		\midrule
		a & a & b & b & \condorcet{c} & \condorcet{c} & a & a & \condorcet{c}\\
		b & b & \condorcet{c} & \condorcet{c} & a & a & \condorcet{c} & \condorcet{c} & a\\
		\condorcet{c} & \condorcet{c} & a & a & b & b & b & b & b\\
		\bottomrule
	\end{array}
	\qquad&
	\begin{tikzpicture}[scale=0.75,baseline=0.2cm]
		\tikzstyle{mynode}=[fill=white,circle,draw,minimum size=1.5em,inner sep=2pt]
		\tikzstyle{mylabel}=[fill=white,circle,inner sep=0.5pt]
		\node[mynode] (a) at (210:1.5) {$a$};
		\node[mynode] (c) at (90:1.5) {$\condorcet{c}$};
		\node[mynode] (b) at (330: 1.5) {$b$};
		\draw[-Latex] (a) edge node[mylabel] {\small$5$} (b);
		\draw[-Latex] (c) edge node[mylabel] {\small$1$} (b);
		\draw[-Latex] (c) edge node[mylabel] {\small$1$} (a);
	\end{tikzpicture}
	\end{align*}
	Thus, since $f$ is a Condorcet extension, we have $f(P_1 + P_2) = \{c\}$, a contradiction.
\end{proof}

\Cref{thm:reinf-9} also holds for ``superset-reinforcement'', which requires that for all profiles $P$ and $P'$, we have $f(P) \cap f(P') \supseteq f(P + P')$ whenever $f(P) \cap f(P') \neq \emptyset$. To see this, we can slightly adapt the proof by noting that superset-reinforcement implies that $f(P_1 + P_2) \subseteq \{a\}$.

Note that the proof of \Cref{thm:reinf-9} only invokes (subset-)reinforcement for pairs of profiles where at least one of them has a Condorcet winner. This is a particularly compelling application of reinforcement since it has the flavor of participation: if $x$ is a winner in a profile and we add some new voters to it who agree that $x$ is the strongest candidate (in the sense that $x$ is the Condorcet winner with respect to the new voters), then $x$ should still be a winner.

\subsection{Proof using Anonymity}

Are 9 voters necessary to obtain the incompatibility between reinforcement and being a Condorcet extension? As it turns out, the result holds even for 8 voters.

\begin{fact}
	\label{fact:reinf-8}
	Every Condorcet extension violates reinforcement when $|N^*| \ge 8$.
\end{fact}

We discovered this fact using a SAT solver,%
\footnote{SAT solving is a computational method that has been leveraged to obtain impossibility theorems in social choice \citep[see, e.g.,][]{GePe17a}. We have also used it to find the proofs of \Cref{thm:reinf-8} and \Cref{thm:reinf-5-neutrality}. To check \Cref{fact:reinf-8}, our Boolean formula did not list all non-anonymous profiles for up to $8$ voters since there are $7^8 > 5$ million of them. Instead, similarly to the approach of \citet[Section 5.2]{BBPS21a}, it was sufficient to list all profiles that can be obtained via voter and candidate permutations from the profiles used in the proof of \Cref{thm:reinf-8}. To encode reinforcement, add a variable $v_{P_1, P_2}$ for each pair $P_1, P_2$ of profiles indicating whether $f(P_1) \cap f(P_2) \neq \emptyset$; then for all $x \in A$ add clauses $(x \in f(P_1) \land x \in f(P_2)) \to v_{P_1, P_2}$ and $(v_{P_1, P_2} \land x \in f(P_1 \cup P_2)) \to x \in f(P_i)$ for $i = 1,2$, and $(v_{P_1, P_2} \land x \in f(P_1) \land x \in f(P_2)) \to x \in f(P_1 \cup P_2)$.}
and we have no human-readable proof for it. The minimal unsatisfiable subsets of clauses (a measure of proof complexity) that we were able to extract argue about hundreds or even thousands of profiles.

However, we can give a proof if we assume anonymity in addition. Anonymity is particularly natural in a variable-electorate model and allows us to apply reinforcement without having to ensure that electorates are disjoint.

\begin{theorem}
	\label{thm:reinf-8}
	 Every anonymous Condorcet extension violates reinforcement when $|N^*| \ge 8$.
\end{theorem}

\begin{proof}
Let $f$ be an anonymous Condorcet extension satisfying reinforcement.
Let $P_1$ be the Condorcet cycle profile:
\[
\begin{array}{ccc}
\toprule
\multicolumn{3}{c}{P_1}\\
\midrule
a & b & c\\
b & c & a\\
c & a & b\\
\bottomrule
\end{array}
\]
Note that by anonymity, we do not need to assign each preference relation to a specific voter $i\in N^*$.
Because $f(P_1) \neq \emptyset$, we may assume without loss of generality that $a \in f(P_1)$.

Consider the following combination of two profiles:
\[
\begin{array}{ccc}
\toprule
\multicolumn{3}{c}{P_1}\\
\midrule
a & b & c\\
b & c & a\\
c & a & b\\
\bottomrule
\end{array}
+
\begin{array}{cccc}
\toprule
\multicolumn{4}{c}{P_2}\\
\midrule
a & a & c & c\\
b & c & a & a\\
c & b & b & b\\
\bottomrule
\end{array}
=
\begin{array}{ccccccc}
\toprule
\multicolumn{7}{c}{P_7}\\
\midrule
a & a & a & b & \condorcet{c} & \condorcet{c} & \condorcet{c}\\
b & b & \condorcet{c} & \condorcet{c} & a & a & a\\
\condorcet{c} & \condorcet{c} & b & a & b & b & b\\
\bottomrule
\end{array}
\]
Throughout the proof, the \condorcet{colored} candidate marks the Condorcet winner of a profile, so $P_7$ has Condorcet winner $c$. From reinforcement and $f$ being Condorcet-consistent, we can infer that $a \not \in f(P_2)$, since otherwise we have $a \in f(P_1) \cap f(P_2)$ and would therefore also have $a \in f(P_7) = \{c\}$, a contradiction.

Similarly, we can deduce from
\[
\begin{array}{ccc}
\toprule
\multicolumn{3}{c}{P_1}\\
\midrule
a & b & c\\
b & c & a\\
c & a & b\\
\bottomrule
\end{array}
+
\begin{array}{cccc}
\toprule
\multicolumn{4}{c}{P_3}\\
\midrule
a & a & b & c\\
c & c & c & a\\
b & b & a & b\\
\bottomrule
\end{array}
=
\begin{array}{ccccccc}
\toprule
\multicolumn{7}{c}{P_8}\\
\midrule
a & a & a & b & b & \condorcet{c} & \condorcet{c}\\
b & \condorcet{c} & \condorcet{c} & \condorcet{c} & \condorcet{c} & a & a\\
\condorcet{c} & b & b & a & a & b & b\\
\bottomrule
\end{array}
\]
that $a \not\in f(P_3)$.

Next, consider
\[
\begin{array}{c}
\toprule
\multicolumn{1}{c}{P_0}\\
\midrule
\condorcet{b}\\
a\\
c\\
\bottomrule
\end{array}
+
\begin{array}{cccc}
\toprule
\multicolumn{4}{c}{P_2}\\
\midrule
a & a & c & c\\
b & c & a & a\\
c & b & b & b\\
\bottomrule
\end{array}
=
\begin{array}{ccccc}
\toprule
\multicolumn{5}{c}{P_4}\\
\midrule
\condorcet{a} & \condorcet{a} & b & c & c\\
b & c & \condorcet{a} & \condorcet{a} & \condorcet{a}\\
c & b & c & b & b\\
\bottomrule
\end{array}
\]
from which we can deduce that $b \not \in f(P_2)$, since otherwise $b \in f(P_0) \cap f(P_2)$ and hence by reinforcement $b \in f(P_4)$, contradicting that $f$ is a Condorcet extension.

Similarly, from
\[
\begin{array}{c}
\toprule
\multicolumn{1}{c}{P_0}\\
\midrule
\condorcet{b}\\
a\\
c\\
\bottomrule
\end{array}
+
\begin{array}{cccc}
\toprule
\multicolumn{4}{c}{P_3}\\
\midrule
a & a & b & c\\
c & c & c & a\\
b & b & a & b\\
\bottomrule
\end{array}
=
\begin{array}{ccccc}
\toprule
\multicolumn{5}{c}{P_6}\\
\midrule
\condorcet{a} & \condorcet{a} & b & b & c\\
c & c & \condorcet{a} & c & \condorcet{a}\\
b & b & c & \condorcet{a} & b\\
\bottomrule
\end{array}
\]
we can deduce that $b \not\in f(P_3)$.

We have shown that $a,b \not\in f(P_2)$ and that $a,b \not\in f(P_3)$. It follows that
\[
f(P_2) = f(P_3) = \{c\}.
\]

Now, because
\[
\begin{array}{cccc}
\toprule
\multicolumn{4}{c}{P_2}\\
\midrule
a & a & c & c\\
b & c & a & a\\
c & b & b & b\\
\bottomrule
\end{array}
+
\begin{array}{cccc}
\toprule
\multicolumn{4}{c}{P_3}\\
\midrule
a & a & b & c\\
c & c & c & a\\
b & b & a & b\\
\bottomrule
\end{array}
=
\begin{array}{cccccccc}
\toprule
\multicolumn{8}{c}{P_9}\\
\midrule
a & a & a & a & b & c & c & c\\
b & c & c & c & c & a & a & a\\
c & b & b & b & a & b & b & b\\
\bottomrule
\end{array}
\]
it follows from reinforcement that $f(P_9) = \{c\}$.

However, because $a \in f(P_1)$ and because $f$ is a Condorcet extension, we get from
\[
\begin{array}{ccc}
\toprule
\multicolumn{3}{c}{P_1}\\
\midrule
a & b & c\\
b & c & a\\
c & a & b\\
\bottomrule
\end{array}
+
\begin{array}{ccccc}
\toprule
\multicolumn{5}{c}{P_5}\\
\midrule
\condorcet{a} & \condorcet{a} & \condorcet{a} & c & c\\
c & c & c & \condorcet{a} & \condorcet{a}\\
b & b & b & b & b\\
\bottomrule
\end{array}
=
\begin{array}{cccccccc}
\toprule
\multicolumn{8}{c}{P_9}\\
\midrule
a & a & a & a & b & c & c & c\\
b & c & c & c & c & a & a & a\\
c & b & b & b & a & b & b & b\\
\bottomrule
\end{array}
\]
that $a \in f(P_9)$ by reinforcement. This contradicts our prior conclusion that $f(P_9) = \{c\}$.
\end{proof}

\subsection{Proof using Anonymity and Neutrality}
\label{sec:neutrality}

When additionally requiring neutrality, the impossibility can be proven using only 5 voters.

\begin{theorem}
	\label{thm:reinf-5-neutrality}
	Every anonymous and neutral Condorcet extension violates reinforcement when $|N^*| \ge 5$.
\end{theorem}

\begin{proof}
	Let $f$ be an anonymous and neutral Condorcet extension satisfying reinforcement. We will consider the following three profiles:
	\[
	\begin{array}{c}
		\toprule
		\multicolumn{1}{c}{P_1}\\
		\midrule
		\condorcet{c}\\
		a\\
		b\\
		\bottomrule
	\end{array}
	\quad
	\begin{array}{cc}
		\toprule
		\multicolumn{2}{c}{P_2}\\
		\midrule
		a & b\\
		b & a\\
		c & c\\
		\bottomrule
	\end{array}
	\quad
	\begin{array}{ccc}
		\toprule
		\multicolumn{3}{c}{P_3}\\
		\midrule
		a & b & c\\
		b & c & a\\
		c & a & b\\
		\bottomrule
	\end{array}
	\]
	
	In profile $P_1$, $c$ is the Condorcet winner, and so $f(P_1) = \{c\}$ because $f$ is a Condorcet extension. 
	Next, we consider $P_2$. We can show that $c \notin f(P_2)$ as follows: If $c \in f(P_2)$, then by reinforcement, we would have $c \in f(P_1 + P_2)$. However, in $P_1 + P_2$, $a$ is the Condorcet winner (throughout this proof, we rearrange voters using anonymity):
	\[
	\begin{array}{ccc}
		\toprule
		\multicolumn{3}{c}{P_1 + P_2}\\
		\midrule
		\condorcet{a} & b & c\\
		b & \condorcet{a} & \condorcet{a}\\
		c & c & b\\
		\bottomrule
	\end{array}
	\]
	Since $f$ is a Condorcet extension, we have $f(P_1 + P_2) = \{a\}$, which contradicts $c \in f(P_1 + P_2)$. Thus, $c \notin f(P_2)$.
	Now, by neutrality and anonymity, and noting that $f(P_2)$ must be non-empty, we have $f(P_2) = \{a, b\}$, as $a$ and $b$ are symmetric in $P_2$.
	Finally, the profile $P_3$ is a Condorcet cycle. By anonymity and neutrality, $f(P_3) = \{a, b, c\}$.
	
	Now let's examine $P_2 + P_3$:
	\[
	\begin{array}{ccccc}
		\toprule
		\multicolumn{5}{c}{P_2 + P_3}\\
		\midrule
		\condorcet{a} & \condorcet{a} & b & b & c\\
		b & b & \condorcet{a} & c & \condorcet{a}\\
		c & c & c & \condorcet{a} & b\\
		\bottomrule
	\end{array}
	\]
	By reinforcement, we must have:
	\[f(P_2 + P_3) = f(P_2) \cap f(P_3) = \{a, b\}\]
	However, in $P_2 + P_3$, $a$ is the Condorcet winner. Thus $f(P_2 + P_3) = \{a\}$, a contradiction.
\end{proof}

The reinforcement axiom was originally introduced by \citet{Youn74a,Youn75a} to give an axiomatic characterization of the class of \emph{scoring rules}. A scoring rule is a social choice function defined by a scoring vector $(s_1, s_2, s_3) \in \mathbb{R}^3$. Given a profile $P$, it then computes the score of each candidate by letting each voter assign $s_1$ points to their top-ranked candidate, $s_2$ points to the second-ranked candidate, and $s_3$ points to their last-ranked candidate. Then the scoring rule selects the candidates with the highest score. A scoring rule is \emph{strictly monotonic} if $s_1 > s_2 > s_3$ and it is \emph{weakly monotonic} if $s_1 \ge s_2 \ge s_3$ and $s_1 > s_3$. Borda's rule is a scoring rule with $s_1 = 2$, $s_2 = 0$, and $s_3 = -2$, where we have chosen the scores so they agree with our earlier definition of the Borda score of a candidate as $\beta_x = \sum_{y \in A} m_{x,y}$. Another classic scoring rule is the \emph{plurality rule} with $s_1 = 1$ and $s_2 = s_3 = 0$, choosing as winners those candidates who are top-ranked most frequently. \citet{Youn75a} showed that if the set of potential voters $N^*$ is infinite, then an anonymous, neutral, and continuous social choice function satisfies reinforcement if and only if it is a scoring rule. Dropping continuity from the axioms, \citet{Youn75a} characterized the class of \emph{composite} scoring rules where additional scoring rules are used to break ties (for example, choosing the plurality winners with the highest Borda score).

Now, because every (composite) scoring rule satisfies anonymity, neutrality, and reinforcement, \Cref{thm:reinf-5-neutrality} implies that every (composite) scoring rule fails Condorcet-consistency on some profile with at most 5 voters. 

This is a new proof of a classic result, first established by \citet[pages clxxvij--clxxviij, translated by \citealp{McHe94a}, pages 137--138]{Cond85a} who constructed an 81-voter profile (\href{https://voting.ml/?profile=30abc-1acb-10cab-29bac-10bca-1cba}{$30abc+acb+10cab+29bac+10bca+cba$}) on which $a$ is the Condorcet winner but every strictly monotonic scoring rule uniquely elects $b$. \citet{Fish84g} gives an example with 7 voters (\href{https://voting.ml/?profile=3abc-2bca-1bac-1cab}{$3abc+2bca+bac+cab$}) that has the same property. For weakly monotonic scoring rules, \citet[Theorem 9.1]{Moul88a} gave a 17-voter example (\href{https://voting.ml/?profile=6abc-3cab-4bac-4bca}{$6abc+3cab+4bac+4bca$}) such that $a$ is the Condorcet winner and every weakly monotonic scoring rule uniquely elects $b$. Moulin suggested that this example is minimal, but in fact an 11-voter example with the same property is minimal (\href{https://voting.ml/?profile=4abc-3bca-2bac-2cab}{$4abc+3bca+2bac+2cab$}, computed by Christian Stricker). Because in these examples, $b$ is the unique scoring winner in each case, we can deduce that no refinement of a (weakly) monotonic scoring rule can be a Condorcet extension. \citet[][Theorem 5]{Webe02a} gives a related example.

The examples discussed in the previous paragraph are all ``universal counterexamples'', while the proof of \Cref{thm:reinf-5-neutrality} is based on reasoning across several profiles. However, by an observation of Florian Grundbacher, the profile $P_2 + P_3$ in that proof  (\href{https://voting.ml/?profile=2abc-1bac-1bca-1cab}{$2abc+bac+bca+cab$}) can also function as a universal counterexample, showing that no scoring rule (including non-monotonic or composite ones) can always \emph{uniquely} select the Condorcet winner. This is because in this profile, $a$ is the Condorcet winner, but $a$ and $b$ are ``rank-indistinguishable'': they both appear twice in first rank, twice in second rank, and once in third rank. Hence, any scoring rule assigns $a$ and $b$ exactly the same score, so that $a$ can never have the uniquely highest score.

Note that these universal counterexamples do not quite imply \Cref{thm:reinf-5-neutrality}, because such an implication would depend on the characterization of \citet{Youn75a} which requires an infinite set of potential voters, while \Cref{thm:reinf-5-neutrality} only requires 5 voters.

\subsection{Lower Bounds}

In this section, we show that the numbers of voters in our incompatibility theorems for reinforcement are minimal. 
To start, \Cref{thm:reinf-5-neutrality} contains a bound of 5 voters. This is the best possible since Black's rule, stable voting, and leximin are anonymous, neutral, and satisfy reinforcement for up to 4 voters. Interestingly, all other Condorcet extensions discussed in \Cref{sec:condorcet-extensions} fail reinforcement even for 4 voters.%
\footnote{Consider the profile \href{https://voting.ml/?profile=1acb-1abc-1bca-1cab}{$acb + (abc + bca + cab)$}, on which every anonymous and neutral Condorcet extension satisfying reinforcement must select $\{a\}$ (this can be seen from the indicated decomposition into two subprofiles). This profile corresponds to Graph D in \Cref{tbl:rules-ordinal}, which indicates that most rules select $\{a, c\}$ on this profile; this is also true for Baldwin. Copeland and Nanson fail reinforcement for the profile \href{https://voting.ml/?profile=2abc-1bac-1bca}{$(abc + bac) + (abc + bca)$}, where they select $\{a,b\}$ on the first subprofile, $\{b\}$ on the second subprofile, but $\{a,b\}$ on the whole profile.}

For our other bounds, we will define an artificial rule that satisfies reinforcement for up to 7 voters and subset-reinforcement for up to 8 voters, which shows that the bounds of \Cref{thm:reinf-8} and \Cref{thm:reinf-9} are the best possible. Given \Cref{thm:reinf-5-neutrality}, this artificial rule necessarily fails neutrality.
It is a refinement of maximin that breaks ties according to an intricate scoring system which we found by solving an integer linear program. This scoring system is non-neutral and specifies for each linear order a number of points that it assigns to its top and its second choice. In addition, for the linear order $abc$, the number of points differs depending on whether the number of voters in the input profile is in $\{2,4,6,8\}$ (Case 1) or not (Case 2). Our rule then returns the maximin winner that obtained the highest score.
\[
\quad
\begin{array}{cc}
	\multicolumn{2}{c}{\text{Case 1}}\\
	\toprule
	a & 11\\
	b & 8\\
	c & 0\\
	\bottomrule
\end{array}
\quad
\begin{array}{cc}
	\multicolumn{2}{c}{\text{Case 2}}\\
	\toprule
	a & 18\\
	b & 13\\
	c & 0\\
	\bottomrule
\end{array}
\quad
\begin{array}{cc}
	\\
	\toprule
	a & 10\\
	c & 7\\
	b & 0\\
	\bottomrule
\end{array}
\quad
\begin{array}{cc}
	\\
	\toprule
	b & 18\\
	a & 11\\
	c & 0\\
	\bottomrule
\end{array}
\quad
\begin{array}{cc}
	\\
	\toprule
	b & 7\\
	c & 0\\
	a & 0\\
	\bottomrule
\end{array}
\quad
\begin{array}{cc}
	\\
	\toprule
	c & 13\\
	a & 5\\
	b & 0\\
	\bottomrule
\end{array}
\quad
\begin{array}{cc}
	\\
	\toprule
	c & 8\\
	b & 0\\
	a & 0\\
	\bottomrule
\end{array}
\]

Using a computer, one can verify that this Condorcet extension satisfies reinforcement up to 7 voters and subset-reinforcement up to 8 voters.\footnote{Code is available at \url{https://gist.github.com/DominikPeters/eeea2456ee4a4cdf2c5aa302b24af7a6}.} It is also anonymous (but not neutral) and satisfies participation as well as monotonicity, which states that improving a winner $x$ in some voters' preference relation while keeping everything else intact should not remove $x$ from the choice set. It fails homogeneity, as one can find a profile with $8$ voters in which $a$ and $c$ are the weak Condorcet winners and $c$ is selected due to the tie-breaking, but doubling the profile makes $a$ the unique winner because different scores are used for $16$ voters.

\section{Participation}
\label{sec:participation}

Throughout this section, we assume that $N^* = \mathbb N$, so the set of potential voters is infinite. This means that it is always possible to copy a profile, and thus, the homogeneity and continuity axioms apply.

The participation axiom demands that voters are never better off by abstaining from an election. For a preference profile $P=({\succ_i})_{i\in N}$ and voter $j \in N$, we write $P_{-j} = ({\succ_i})_{i\in N \setminus\{j\}}$ for the profile obtained from $P$ by deleting $j$. For resolute rules, the definition of the participation axiom is straightforward: voter $j$ should weakly prefer the outcome when $j$ is present to the outcome when $j$ is absent.
\begin{definition}[Resolute participation, \citealp{Moul88b}]
	A resolute social choice function satisfies \emph{resolute participation} if for all profiles $P$ and $j\in N$ with $f(P) = \{x\}$ and $f(P_{-j}) = \{y\}$, we have $x \succeq_j y$.
\end{definition}

\citet{Moul88b} has shown that when there are at least four candidates, then every resolute Condorcet extension violates resolute participation.
This incompatibility holds for 12 or more voters \citep{BGP16c}. 
On the other hand, \citet{Moul88b} proved that for three candidates, maximin with a fixed tie-breaking order satisfies resolute participation.

For irresolute social choice functions, defining participation requires us to compare the sets of candidates $f(P)$ and $f(P_{-j})$ with respect to the preference relation $\succ_j$. This can be done in various ways. We will focus on the following well-studied variant that evaluates sets ``optimistically'', that is, by comparing the most-preferred candidates contained in each set. For a non-empty set $A' \subseteq A$ of candidates, we write $\max_{\succ_i} A'$ for the most-preferred candidate in $A'$ according to $\succ_i$, so that $\max_{\succ_i} A' \succeq_i x$ for all $x \in A'$.
\begin{definition}[Optimist participation]
	A social choice function satisfies \emph{optimist participation} if for all profiles $P$ and all $j\in N$, $\max_{\succ_j} f(P) \succeq_j  \max_{\succ_j} f(P_{-j})$.
\end{definition}
For four or more candidates, \citet{JPG09a} showed that no Condorcet extension satisfies optimist participation \citep[see, also,][Thm.~6]{BGP16c}. For three candidates, maximin satisfies optimist participation. 
To see this, let $P$ be a profile and $j \in N$ be a voter with preference $\succ_j$. Let $g$ be the resolute social choice function obtained from maximin by breaking ties using $\succ_j$ as tie-breaking order. Then $g$ satisfies resolute participation \citep{Moul88b}. Thus,
\[
\textstyle
\max_{\succ_j} f(P) = g(P) \succeq_j  g(P_{-j}) = \max_{\succ_j} f(P_{-j}),
\]
so that maximin satisfies optimist participation.%
\footnote{\label{fn:fishburn}In fact, maximin satisfies a stronger notion of participation based on the Fishburn preference extension \citep{Gard79a}. In particular, writing $X = f_{\text{maximin}}(P)$ and $Y = f_{\text{maximin}}(P_{-j})$, it can be shown that we always have $X \succ_j Y \setminus X$ and $X \setminus Y \succ_j Y$, by applying standard results about the Fishburn preference extension (\citealp[Theorem 3.4]{ErSa09a}; \citealp[Section 3.3(ii)]{BSS19a}) and using the fact that maximin satisfies resolute participation for all ways of breaking ties using a fixed tie-breaking ordering.}

To compare our results with existing results, let us define two additional participation-related conditions. A social choice function satisfies \emph{positive involvement} \citep{Pere01a} if whenever a voter joins a profile in which that voter's top candidate wins, then this candidate still wins. Formally, for all profiles $P$ and all $j\in N$ with $x$ the top-ranked candidate of $j$, we have that if $x \in f(P_{-j})$ then $x \in f(P)$. \citet{Pere01a} also defined an analogous notion of \emph{negative involvement}, which demands that if an agent's least-favorite candidate loses, then it still loses when the agent joins the electorate. We define an alternative version of this condition, \emph{singleton negative involvement}, which says that if $f(P_{-j}) \neq \{x\}$ and $x$ is bottom ranked by $j$, then $f(P) \neq \{x\}$.%
\footnote{While similar in spirit, singleton negative involvement and negative involvement are logically incomparable. However, for resolute social choice functions, they are equivalent.}
For three candidates, it turns out that the conjunction of positive involvement and singleton negative involvement is equivalent to optimist participation, which one can verify by case analysis on all possible combinations of sets $f(P_{-j}) \subseteq A$ and $f(P) \subseteq A$.
\begin{proposition}
	\label{thm:optimist-equivalence}
	A social choice function satisfies optimist participation if and only if it satisfies positive involvement and singleton negative involvement.
\end{proposition}

\citet{Pere95a} noted that in order to satisfy positive involvement, a Condorcet extension must refine what \citet{Holl24a} later called the defensible set.
The \emph{defensible set} of a profile $P$ is
\[
f_{\textup{defensible}}(P) = \{ x \in A : \text{for all $y \in A$, there is $z \in A$ with $m_{z, y} \ge m_{y, x}$} \}.
\]
We include a proof for completeness, noting that it is similar to an argument that was already used by \citet[Claim (3)]{Moul88b}.%
\footnote{\label{fn:defensible-set-converse-fails}The ``converse'' does not hold. Indeed, the strict Nanson rule and Baldwin's rule refine the defensible set (\Cref{thm:hasse}) but fail positive involvement even for three candidates: on the profile \href{https://voting.ml/?profile=2acb-2bac-2cba}{$2acb+2bac+2cba$}, both rules select $\{a, b, c\}$; after adding a $cab$ voter, we get the profile \href{https://voting.ml/?profile=2acb-2bac-1cab-2cba}{$2acb+2bac+cab+2cba$} where both rules select $\{a\}$, violating positive involvement.}

\begin{lemma}[\citealp{Holl24a}, Lemma 2(1); \citealp{Pere95a}, Lemma 3]
	\label{thm:pi-refines-defensible-set}
	Let $f$ be a Condorcet extension that satisfies positive involvement. Then $f$ is a refinement of $f_{\textup{defensible}}$, the defensible set.
\end{lemma}
\begin{proof}
	Let $P$ be a profile. Let $x \notin f_{\text{defensible}}(P)$, so that there exists $y\in A$ with $m_{y,x}> m_{z,y}$ for all $z\in A$.
	Recall that all margins of a profile must have the same parity, so $m_{y,x}> m_{z,y}$ implies that we even have $m_{y,x} > m_{z,y}+1$ for all $z\in A$. 
	We need to show that $x \notin f(P)$. Assume for a contradiction that $x\in f(P)$. Add $m_{y,x}-1$ voters with rankings $x\succ y\succ \dots$ to $P$ to obtain the profile $P^*$. By positive involvement, $x$ has to remain chosen by $f$ after each addition of one voter, and thus, $x \in f(P^*)$. However, $y$ is the Condorcet winner in  $P^*$, because we have margins $m_{y,x}^* = 1$ and $m_{y,z}^* = -m_{z,y} - 1 + m_{y,x}>0 $. Thus, because $f$ is a Condorcet extension, we have $x \notin f(P)$, a contradiction.
\end{proof}

\subsection{Participation Requires Refining Maximin}

As noted above, maximin satisfies optimist participation for three candidates. 
We will now show that maximin is, in fact, pre-eminent among Condorcet extensions having this property: all such Condorcet extensions must be refinements of maximin under the mild additional assumption of homogeneity.%
\footnote{Our proof of \Cref{thm:maximin-refinement} actually uses the weaker ``inclusion homogeneity'' \citep[Definition 3.6]{HoPa23a} requiring only $f(P)\subseteq f(2P)$. In fact, we only need that for every profile $P$, there exists \emph{at least one} profile $2P$ that is a 2-fold copy of $P$ and such that $f(P) \subseteq f(2P)$. In the absence of anonymity, this is a weaker notion. Indeed, there are some non-anonymous refinements of maximin that satisfy participation and this weakened homogeneity, but not full homogeneity. An example is maximin with ties broken in favor of the most-preferred candidate of the voter $i \in N \subseteq N^* = \mathbb{N}$ with the smallest label. This satisfies homogeneity as long as the copy $2P$ has the same smallest-labelled voters as in $P$.}

\begin{theorem}\label{thm:maximin-refinement}
	Let $f$ be a homogeneous Condorcet extension that satisfies optimist participation. Then $f$ is a refinement of maximin.
\end{theorem}

Our proof will use \Cref{thm:pi-refines-defensible-set} to deduce that $f$ must refine the defensible set. It will then use the extra strength of optimist participation compared to positive involvement (in particular, singleton negative involvement, see \Cref{thm:optimist-equivalence}) to narrow down the space of candidates that $f$ can select to conclude that $f$ must refine maximin.

\begin{proof}
	On profiles that induce a Condorcet winner, $f$ refines maximin because it is a Condorcet extension.
	On profiles whose margin graph is a cycle with three equal non-negative margins, maximin selects all three candidates, and so $f$ trivially refines maximin.
	All other profiles induce margin graphs that match, up to renaming the candidates, one of the four cases depicted in \Cref{fig:MaximinRefinementFigure}.%
	\footnote{Comparing to the graphs shown in \Cref{tbl:rules-ordinal}, we see that Case 1 corresponds to Graphs K and L; Case 2 corresponds to Graphs E and F; Case 3 (left) corresponds to Graphs G and H; Case 3 (right) corresponds to Graphs I and J; Case 4 corresponds to Graphs C and D.}

	\begin{figure}[t]
		\centering
		\includegraphics[width=\textwidth]{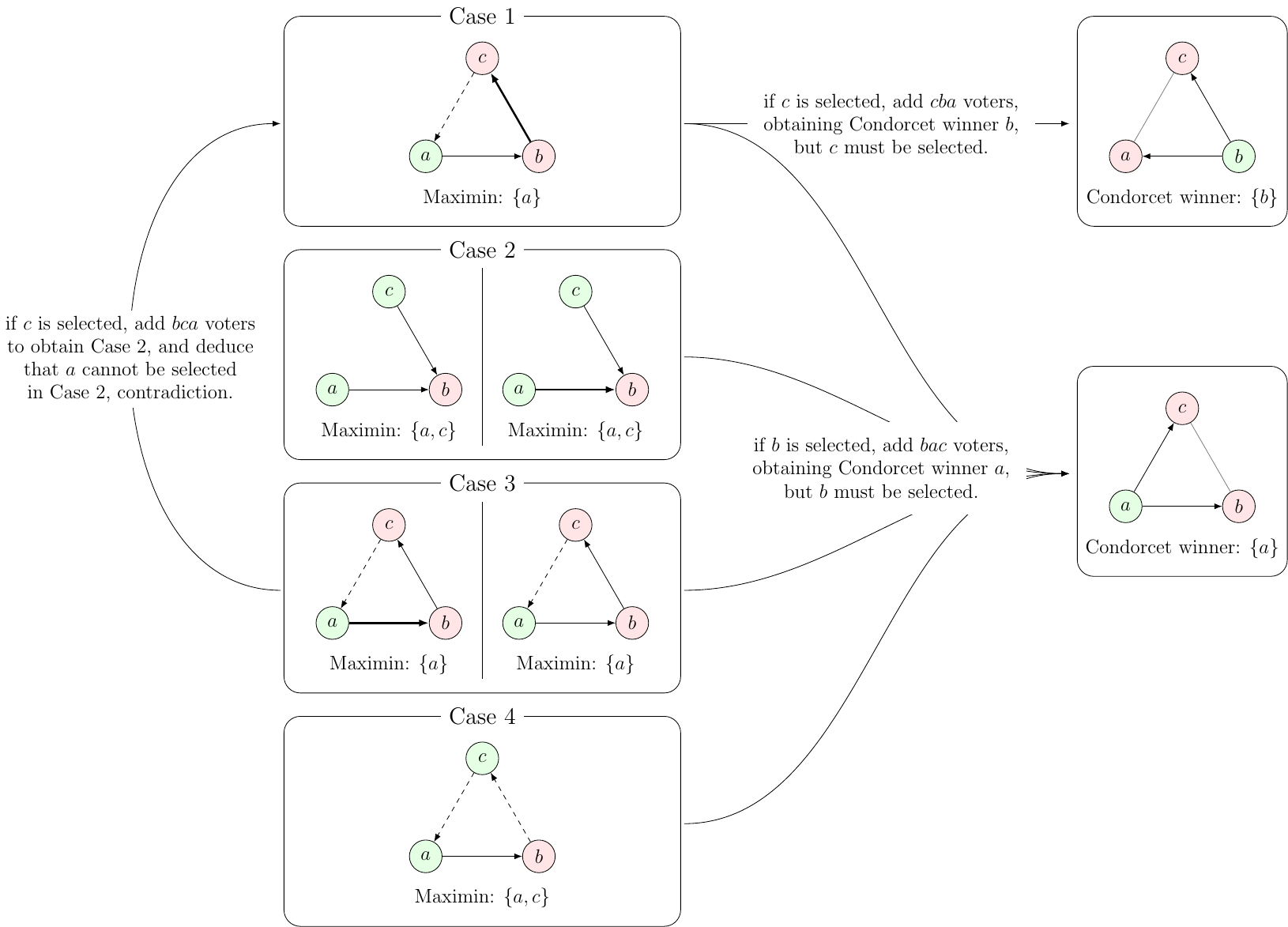}
		\caption{A proof sketch of \Cref{thm:maximin-refinement}. In each panel, arrows from \(x\) to \(y\) indicate that \(x\) has a non-negative majority margin against \(y\): bold arrows represent the largest margins, normal arrows smaller but strictly positive margins, and dashed arrows the smallest margins, which may be zero. Between panels, the arrow from $A$ to $B$ explains why the choice of a candidate in panel $A$ would contradict the choices established in panel $B$.
		}%
		\label{fig:MaximinRefinementFigure}
	\end{figure}

	The defensible set coincides with the maximin winners in Cases 1, 2, and 4. Because $f$ satisfies positive involvement (by \Cref{thm:optimist-equivalence}), it refines the defensible set by \Cref{thm:pi-refines-defensible-set}, and therefore $f$ refines maximin in Cases 1, 2, and 4. (\Cref{fig:MaximinRefinementFigure} makes the argument of \Cref{thm:pi-refines-defensible-set} explicit by explaining what voters to add to arrive at a Condorcet winner and a contradiction to participation.)
	
	It remains to consider Case 3, which concerns profiles $P$ that, up to renaming candidates, induce a cycle with $0\le m_{c,a}< m_{b,c} \le m_{a,b}$. In $P$, the unique maximin winner is $a$. Thus, our task is to show that $f(P) = \{a\}$. We will show this by connecting this case to Case 1. Because $f_{\textup{defensible}}(P) = \{a,c\}$, we get that $b\notin f(P)$ from \Cref{thm:pi-refines-defensible-set}. Assume for a contradiction that $c\in f(P)$. Let $P' = 2P$ be a profile consisting of two copies of $P$. By homogeneity, we obtain $c \in f(P')$.
	Let $j= (m'_{a,b}-m'_{b,c})/2 +1$.
	Adding $j$ voters to $P'$ with preferences $b\succ c\succ a$ yields a profile $P^*$ with margins 
	\[ m^*_{c,a}= m'_{c,a}+j, \qquad m_{a,b}^* =m'_{a,b}-j, \quad \text{and} \quad m_{b,c}^*= m'_{b,c} + j. \] 
	Crucially, we have $m'_{c,a} + 2 <  m'_{b,c}$ and by choice of $j$ we have $m_{c,a}^* < m_{a,b}^* < m_{b,c}^*$.
	Thus, $P^*$ is a profile matching Case 1. Following our prior analysis, this implies $f(P^*) =\{a\}$. This is a contradiction to optimist participation because $\{a\}$ is a worse outcome for voters with preferences $b\succ c\succ a$ than an outcome containing $c$ as a winner. (Note this step uses just singleton negative involvement.)
\end{proof}

\Cref{thm:maximin-refinement} can be reinterpreted for resolute social choice functions, as optimist and resolute participation are equivalent for resolute social choice functions. Thus, we get the following corollary for resolute Condorcet extensions.

\begin{corollary}
	Let $f$ be a homogeneous and resolute Condorcet extension that satisfies resolute participation. Then $f$ is a refinement of maximin.
\end{corollary}

Inspecting the rules shown in the Hasse diagram of \Cref{fig:hasse}, we see that \Cref{thm:maximin-refinement} immediately implies that many rules fail optimist participation.
However, there are some important rules that do satisfy optimist participation (and also the strengthened version described in \Cref{fn:fishburn}): these are maximin itself (and the many other rules equivalent to it for three candidates), stable voting, Nanson's rule, and leximin. In addition, if we use a fixed tie-breaking order to make any of these rules resolute, they satisfy (optimist/resolute) participation. On the other hand, the strict Nanson rule fails optimist participation (and even positive involvement, see \Cref{fn:defensible-set-converse-fails}) despite being an anonymous and neutral maximin refinement.

\Cref{thm:maximin-refinement} shows that every social choice function satisfying certain axioms must refine maximin. A theorem with the same conclusion was recently obtained by \citet{HoPa23a}. They introduce the axioms of \emph{weak positive responsiveness}, which requires that if $z \in f(P)$ and if $P'$ is obtained from $P$ through some voter changing her preferences from $x \succ y \succ z$ to $z \succ x \succ y$, then $f(P') = \{z\}$, and the axiom of \emph{block preservation} which requires that if $Q$ is a profile in which each of the $3! = 6$ possible preferences occurs exactly once, then $f(P) \subseteq f(P + Q)$.
\begin{theorem}[\citealp{HoPa23a}, Theorem 3.8 and Proposition 3.10]
	\label{thm:HoPa23a}
	Let $f$ be an anonymous, neutral, and homogeneous Condorcet extension that satisfies positive involvement, weak positive responsiveness, and block preservation. Then $f$ is a refinement of maximin.
\end{theorem}
Note that our \Cref{thm:maximin-refinement} does not require the axioms of anonymity, neutrality, weak positive responsiveness, and block preservation but instead uses singleton negative involvement (see \Cref{thm:optimist-equivalence}).

In the following sections, we will characterize several rules as the unique maximin refinements satisfying certain axioms. Each of these results can be combined with either the axioms of \Cref{thm:maximin-refinement} (or those of \Cref{thm:HoPa23a}) to obtain axiomatic characterizations within the class of all social choice functions.

\subsection{Characterization of Leximin}
\label{sec:leximin}

Leximin refines maximin and is remarkably resolute because it only returns a tie when two candidates are ``locally indistinguishable'' within the margin graph (see \Cref{fn:leximin definition}). This allows it to satisfy a demanding monotonicity-style axiom called positive responsiveness. We say that a profile $P'$ is obtained from a profile $P$ by \emph{improving $x$ relative to $y$} if the ranking of all candidates stays fixed except that one or more voters who rank $y$ immediately above $x$ in $P$ instead rank $x$ immediately above $y$ in $P'$. A social choice function $f$ satisfies \emph{positive responsiveness} (\citealp{May52a}; \citealp{Barb77b}, Definition 7) if whenever $x \in f(P)$ and $P'$ is obtained from $P$ by improving $x \in f(P)$ relative to $y \in A$, we have $f(P') = \{x\}$. We show that leximin is the unique maximin refinement that satisfies this property among maximin refinements that are neutral and that are pairwise. A social choice function $f$ is \emph{pairwise} if it only depends on the margin graph, i.e., for all $P$ and $P'$ such that $m_{x,y}(P) = m_{x,y}(P')$ for all $x,y\in A$, we have $f(P)=f(P')$.

\begin{lemma}\label{lem:leximinChar}
	Leximin is the only neutral and pairwise refinement of maximin that satisfies positive responsiveness.
\end{lemma}
\begin{proof}
	We first show that leximin satisfies all axioms.
	It is neutral and pairwise by definition. To see that it satisfies positive responsiveness, note that if $P'$ is obtained from $P$ by moving up $x \in f_{\text{leximin}}(P)$ then the margins of $x$ only improve, i.e., $m^x(P') \triangleright m^x(P)$, while for any $z\neq x$ the margins can only become worse, i.e., $m^z (P)\trianglerighteq m^z(P')$. (See \Cref{fn:leximin definition} for the definition of the relation $\triangleright$.) Let $z\in A\setminus\{x\}$. Then because $x$ is selected by leximin, we have $m^x(P)\trianglerighteq m^z (P)$. Thus, $m^x(P') \triangleright m^x(P)\trianglerighteq m^z (P) \trianglerighteq m^z(P')$. Hence $z \notin f_{\text{leximin}}(P')$ and so $f_{\text{leximin}}(P') =\{x\}$ as desired.

	For uniqueness, let $f$ be a maximin refinement satisfying neutrality, pairwiseness, and positive responsiveness. By pairwiseness, note that $f$ is also anonymous. We need to show that $f = f_{\text{leximin}}$. By pairwiseness, it suffices to consider each possible margin graph, find some profile $P$ that induces it, and show that $f(P) = f_{\text{leximin}}(P)$.
	So consider any margin graph. If maximin selects a unique winner when given a profile inducing that margin graph, we are done since then all maximin refinements select the same choice set. All other margin graphs correspond to Graphs A to F in \Cref{tbl:rules-ordinal}.

	For Graph A (the Condorcet cycle with three edges each with positive weight $t$) take $P = t \cdot (abc + bca + cab)$ and for Graph B (all pairs of candidates are majority tied) take the profile $P$ consisting of one copy of each of the six preference orders. Anonymity and neutrality imply $f(P) = \{a,b,c\}$ which equals $f_{\text{leximin}}(P)$.
	
	For Graphs C and D, we have a majority graph with $0 \le m_{c,a} = m_{b,c} < m_{a,b}$. Let $P$ be a profile inducing this graph (which exists by McGarvey's theorem). Let $t = \frac12 (m_{a,b} - m_{b,c}) > 0$. We may assume that there are at least $t$ many voters with preferences  $c\succ a\succ b$ in $P$ (otherwise we can add those voters to $P$ and also add $t$ voters with opposite preferences $b \succ a \succ c$ which leads to a profile inducing the same margin graph). Let $P'$ be the profile obtained from $P$ by selecting $t$ voters with preferences $c \succ a\succ b$ and replace their preference by $c\succ b\succ a$. Then on $P'$, all majority margins are equal, i.e., $0 \le m_{c,a}(P') =m_{a,b}(P') = m_{b,c}(P')$. This corresponds to Graphs A or B, and thus $f(P')=\{a,b,c\}$. Note that $P$ can be viewed as being obtained from $P'$ by improving $a \in f(P')$ relative to $b$. Thus, by positive responsiveness, we obtain $f(P)= \{a\} = f_{\text{leximin}}(P)$.
	
	For Graph E, where the two edges both have weight $2t$, consider the profile $P = t \cdot (acb + cab)$ inducing it.
	Note that $b\notin f_{\text{maximin}}(P)$ and thus $b \notin f(P)$.
	By anonymity and neutrality, we obtain $f(P)=\{a,c\}= f_{\text{leximin}}(P)$. 
	
	For Graph F,  we have $m_{a,b} > m_{c,b} > m_{c,a} = 0$. Let $t = \frac12(m_{a,b} - m_{c,b}) > 0$ and let $P$ be a profile inducing Graph F with at least $t$ voters having preference $a\succ b\succ c$. By changing $t$ of these votes to $b\succ a\succ c$, we reach a profile $P'$ on which the new majority margins satisfy $m_{a,b}(P') = m_{c,b}(P')$, which is of type Graph E, and thus $f(P') = \{a,c\}$. Note that $P$ can be viewed as being obtained from $P'$ by improving $a \in f(P')$ relative to $b$. Thus, by positive responsiveness, we obtain $f(P) = \{a\} = f_{\text{leximin}}(P)$, concluding the proof.
\end{proof}

As we explained, combined with \Cref{thm:maximin-refinement}, we obtain an axiomatic characterization.

\begin{corollary}
	Leximin is the only homogeneous, neutral, and pairwise Condorcet extension that satisfies optimist participation and positive responsiveness.
\end{corollary}

\subsection{Characterization of Nanson's rule}

Leximin is very decisive and only reports ties in very specific situations. However, one may argue that leximin neglects to select some candidates that deserve to win. In particular, it may fail to select a candidate even if it is ``almost'' a Condorcet winner in the sense that, in pairwise majority comparisons, it is never beaten but does beat some other candidates. For example, in Graph F of \Cref{tbl:rules-ordinal}, candidate $c$ meets this description but is not selected by leximin, which selects only $\{a\}$.  In contrast, Nanson's rule selects $\{a,c\}$. Note that Graph F is the only case where leximin and Nanson's rules differ.
Graphs D and E also have intermediate Condorcet winners ($\{a\}$ and $\{a,c\}$, respectively).

The concept of candidates that are almost Condorcet winners in this sense was first considered by \citet{BuWe79a} and \citet{Gehr83a}. They were later termed ``intermediate Condorcet winners'' by \citet{BaBo24a}. 
Formally, a candidate $x$ is an \emph{intermediate Condorcet winner} if $m_{x,y}\ge 0$ for all $y\in A$, and at least one of these inequalities is strict. \citet{BaBo24a} say that $f$ is a \emph{strong Condorcet extension} if $f(P)$ is equal to the set of intermediate Condorcet winners whenever this set is non-empty. Note that whenever a candidate is a Condorcet winner, it is also the unique intermediate Condorcet winner. In Graph F in \Cref{tbl:rules-ordinal}, $a$ and $c$ are intermediate Condorcet winners, and thus a strong Condorcet extension must select both of them. Since leximin only selects $a$ it is not a strong Condorcet extension (the same is true for stable voting). However, Nanson's rule is a strong Condorcet extension.%
\footnote{For more than 3 candidates, Nanson's rule is not a strong Condorcet extension. In the profile \href{https://voting.ml/?profile=2ABCD-4ADBC-1BCDA-1CDAB-4CDBA}{$2abcd+4adbc+1bcda+1cdab+4cdba$}, $a$ is the unique intermediate Condorcet winner but Nanson selects $\{c\}$.}

Given our previous result about leximin, Nanson's rule must violate positive responsiveness. However, it does satisfy a slight weakening that is concerned with voters who rank two winning candidates next to each other; if such a voter flips those candidates, then the tie should be broken in favor of the strengthened candidate.
Thus, a social choice function $f$ satisfies \emph{tie-break positive responsiveness} if whenever $x, y \in f(P)$ and $P'$ is obtained from $P$ by improving $x$ relative to $y$, then $f(P') = \{x\}$. This is a weaker property than positive responsiveness.

\begin{lemma}
	Nanson's rule is the only neutral and pairwise refinement of maximin that satisfies tie-break positive responsiveness and is a strong Condorcet extension.
\end{lemma}
\begin{proof}
	We first show that Nanson's rule satisfies these axioms.
	It refines maximin by \Cref{thm:hasse} and is pairwise and neutral as Borda scores only depend on the majority margins and are neutral. 
	That Nanson's rule is a strong Condorcet extension can be seen from \Cref{tbl:rules-ordinal}.
	Tie-break positive responsiveness is satisfied due to the following case distinction: Let $P$ be a profile. If $|f_{\text{Nanson}}(P)| = 1$, then tie-break positive responsiveness does not apply to $P$. If $|f_{\text{Nanson}}(P)| = 2$, say $f_{\text{Nanson}}(P) = \{a,c\}$, then $a$ and $c$ have positive Borda score in the first round and the resulting majority comparison is a tie. When we improve, say, $a$ relative to $c$, this increases the Borda score of $a$ and decreases the one of $c$ (and keeps the Borda score of $b$ the same). Thereafter, $a$ may be the unique candidate with positive Borda score in the first round, or $a$ and $c$ both retain positive Borda score, but then $a$ wins the majority comparison against $c$. Hence, in either case, $f_{\text{Nanson}}(P') = \{a\}$, as required by tie-break positive responsiveness. If $|f_{\text{Nanson}}(P)| = 3$, then the rule stopped without eliminating any candidates, which means that all candidates have a Borda score of zero in $P$. Hence, improving some $x$ relative to any $y$ makes $x$ the unique candidate with strictly positive Borda score, and thus $f_{\text{Nanson}}(P') = \{x\}$, as required by tie-break positive responsiveness.

	For uniqueness, let $f$ be neutral and pairwise maximin refinement that is a strong Condorcet extension and that satisfies tie-break positive responsiveness. We go through all possible margin graphs and show that $f$ coincides with Nanson's rule.
	Whenever maximin chooses only one candidate, it is clear that all refinements return the same choice set.
	Considering \Cref{tbl:rules-ordinal}, we see that only Graphs A to F remain.
	
	For Graphs A and B (Condorcet cycles with three equally and non-negatively weighted edges), pairwiseness and neutrality imply that $f(P) = \{a,b,c\} = f_{\text{Nanson}}(P)$ as in the proof of \Cref{lem:leximinChar}.
	The argument for Graphs C and D is identical to the corresponding argument in the proof of \Cref{lem:leximinChar} because tie-break positive responsiveness suffices. 
	For Graphs E and F, strong Condorcet-consistency requires that $f(P) =\{a,b\} = f_{\text{Nanson}}(P)$. This concludes the proof that $f$ is equal to Nanson's rule.
\end{proof}

As before, combined with \Cref{thm:maximin-refinement}, we obtain an axiomatic characterization.

\begin{corollary}
	Nanson's rule is the only homogeneous, neutral, and pairwise strong Condorcet extension that satisfies optimist participation and tie-break positive responsiveness.
\end{corollary}

\subsection{Characterization of Maximin}

Lastly, we characterize maximin itself. We do so using the continuity axiom, which has been used in other characterizations to rule out tie-breaking mechanisms \citep[e.g.,][]{Youn75a}.\footnote{To see why continuity rules out tie-breaking, suppose candidates $a$ and $b$ are tied, with the tie broken in favor of $a$. When copying the profile $n$ times, there will still be a tie, broken in favor of $a$ (assuming homogeneity). However, if we add a profile $P'$ where $b$ is a clear winner, then in $nP + P'$, candidate $b$ will be the strongest. This is a failure of continuity: $f(nP + P') = \{b\} \not\subseteq \{a\} = f(P)$.}
One can view leximin and Nanson as applying a certain tie-breaking mechanism to maximin, and indeed these two rules fail continuity.
In fact, \emph{all} of maximin’s strict refinements fail continuity.
This is true for any number of candidates and may thus be of independent interest.

\begin{lemma}
	\label{thm:continuous-refinements-of-maximin}
	Maximin is continuous.
	If $f$ is continuous and a refinement of maximin, then $f$ is equal to maximin.
\end{lemma}

\begin{proof}
	In this proof, a candidate $x^* \neq x$ is said to be an \emph{opponent} of $x$ w.r.t.\ profile $P$ if $m_{x,x^*}(P) \le m_{x,z}(P)$ for all $z \neq x$. With this terminology, the definition of the maximin rule becomes $f_{\text{maximin}}(P) = \{ x \in A : m_{x, x^*}(P) \ge m_{y, y^*}(P) \text{ for all $y \in A$}\}$.
	
	We first claim the following: 
	\begin{equation}
		\label{eq:continuous-opponent-claim}
		\parbox{.85\textwidth}{for any profiles $P$ and $P'$ and any candidate $x$, if $x^*$ is an opponent of $x$ w.r.t.\ $P$ and among these minimizes $m_{x,x^*}(P')$, then for all large enough $n$, $x^*$ is also an opponent of $x$ w.r.t.\ $nP + P'$.}
	\end{equation}
	To prove this, we have to show that $m_{x,x^*}(nP + P') \le m_{x,z}(nP + P')$ for all $z \neq x$. First, let $z$ not be an opponent of $x$. Then, $m_{x,x^*}(P) < m_{x,z}(P)$, and hence for large enough $n$, we have
	\[
		m_{x,x^*}(nP + P') = n m_{x,x^*}(P) + m_{x,x^*}(P') < n m_{x,z}(P) + m_{x,z}(P') = m_{x,z}(nP + P').
	\]
	Otherwise, $z$ is an opponent of $x$ and hence $m_{x,x^*}(P) = m_{x,z}(P)$. By choice of $x^*$, we obtain 
	\[
		m_{x,x^*}(nP + P') = n m_{x,x^*}(P) + m_{x,x^*}(P') \le n m_{x,z}(P) + m_{x,z}(P') = m_{x,z}(nP + P'), 
	\]
	proving the claim.
	
	To show that maximin satisfies continuity, let $P$ and $P'$ be two profiles and write $W = f_{\text{maximin}}(P)$. For each $z \in A$, denote by $z^*$ an opponent of $z$ w.r.t.\ $P$ which among them minimizes $m_{z,z^*}(P')$. If $W = A$, then trivially $f_{\text{maximin}}(nP + P') \subseteq W$ and we are done. Otherwise, let $x \in W$ and $y \notin W$, and so $m_{x,x^*}(P) > m_{y, y^*}(P)$. By choosing $n$ large enough, we also obtain that $m_{x,x^*}(nP + P') > m_{y, y^*}(nP + P')$. By \eqref{eq:continuous-opponent-claim}, $y$ is not chosen by maximin in $nP + P'$. 
	Hence, for $n$ large enough, we have $f_{\text{maximin}}(nP + P') \subseteq W$.
	
	For the second claim, let $f$ be a refinement of maximin that is continuous. Let $P$ be a profile and write $W = f_{\text{maximin}}(P)$; we need to show that $f(P) = W$.
	If $|W| = 1$, this is true since $f$ refines maximin. Otherwise, write $W = \{x_1, \dots, x_k\}$ with $k \ge 2$. Without loss of generality, it suffices to show that $x_1 \in f(P)$. 

	Suppose first that $m_{x_i, z}(P) \ge 0$ for all $x_i \in W$ and all $z \in A$. Then for all $x_i, x_j \in W$, we have $m_{x_i, x_j}(P) \ge 0$ and $m_{x_j, x_i}(P) \ge 0$. Since $m_{x_i, x_j}(P) = -m_{x_j, x_i}(P)$, it follows that $m_{x_i, x_j}(P) = 0$.
	Now let $P'$ be a profile consisting of a single voter who ranks $x_1$ on top. For every $n \ge 1$, the profile $n P + P'$ has $x_1$ as its Condorcet winner because it now strictly beats all other candidates. Because $f$ is a Condorcet extension (being a maximin refinement) and because it is continuous, we have $\{x_1\} = f(nP + P') \subseteq f(P)$ and therefore $x_1 \in f(P)$, as desired.
	
	Otherwise, $m_{x_k, y}(P) < 0$ for some $x_k \in W$ and some $y \neq x_k$. For each $i$, let $x_i^*$ be an opponent of $x_i$ w.r.t.\ $P$. Becaue $m_{x_k, y}(P) < 0$, we have $m_{x_k, x_k^*}(P) < 0$. Thus by definition of maximin, we have $m_{x_i, x_i^*}(P) < 0$ for all $x_i \in W$. Using McGarvey's theorem, we now construct a profile $P'$ with the following margin graph: for each $x_i \in W \setminus \{x_1\}$, set $m_{x_i, x_i^*}(P') = -2$ and $m_{x_i^*, x_i}(P') = 2$; all other majority margins are set to zero. To see that this is well-defined (and we do not assign different weights to the same edge), note that if we had $(x_i, x_i^*) = (x_j^*, x_j)$ for some $i$ and $j$, then $m_{x_i, x_i^*}(P) = m_{x_j^*, x_j}(P) = -m_{x_j, x_j^*} > 0$, contradicting $m_{x_i, x_i^*}(P) < 0$.
	Now, for each $x_i \in W$, the margin $m_{x_i, x_i^*}(P') = -2$ is the smallest margin appearing in $P'$ and thus $x_i^*$ is an opponent $x_i$ w.r.t.\ $P$ which minimizes $m_{x_i, x_i^*}(P')$.
	Thus for large enough $n$, due to \eqref{eq:continuous-opponent-claim}, we find that for every $x_j \in W$,  the candidate $x_j^*$ is an opponent of $x$ w.r.t.\ $nP + P'$.
	Note that because maximin satisfies continuity, we have $f_{\text{maximin}}(nP + P') \subseteq f_{\text{maximin}}(P) = W$ for large enough $n$.
	In addition, for all $x_j \neq x_1$, we have 
	\[
	m_{x_1, x_1^*}(nP + P') = nm_{x_1, x_1^*}(P) + 0 = nm_{x_j, x_j^*}(P) >  nm_{x_j, x_j^*}(P) - 2 = m_{x_j, x_j^*}(nP + P'),
	\]
	and hence $x_j \notin f_{\text{maximin}}(nP + P')$. Thus $f_{\text{maximin}}(nP + P') = \{x_1\}$.
	Because $f$ is a maximin refinement and continuous, we obtain $\{x_1\} = f(nP + P') \subseteq f(P)$ for $n$ large enough, and therefore $x_1 \in f(P)$, as desired.
\end{proof}

Again, combined with \Cref{thm:maximin-refinement}, we obtain an axiomatic characterization.

\begin{corollary}
	\label{cor:maximin-characterization}
	Maximin is the only homogeneous and continuous Condorcet extension that satisfies optimist participation.
\end{corollary}

The axioms in \Cref{cor:maximin-characterization} are independent: leximin satisfies all axioms except continuity, Borda's rule satisfies all axioms except Condorcet-consistency, and the top cycle satisfies all axioms except optimist participation. 
Without homogeneity, we can construct an artificial rule that is identical to maximin, except for the profile \href{https://voting.ml/?profile=3ABC-3BCA-2CAB-1ACB}{$3abc+3bca+2cab+1acb$}, where the rule returns $\{a,c\}$ instead of $\{a\}$.
Clearly, homogeneity is violated (doubling this profile leads to outcome $\{a\}$). All other properties are inherited from maximin. Checking that optimist participation is satisfied can be done by case analysis.

\section{Conclusion}

We have investigated whether the search for a desirable Condorcet extension becomes easier when focussing on the special case of three candidates. Our results highlight the maximin rule and two of its refinements (Nanson's rule and leximin) as being particularly robust to common criticisms of Condorcet extensions. Indeed, we showed that they are axiomatically characterized by their immunity to the no-show paradox, together with other desirable properties such as positive responsiveness. These conclusions could motivate advocating for their adoption in real-world elections.

\citet{Nans82a} gave a rather simple description of his rule for three candidates: each voter assigns 2 points to his most preferred candidate and 1 point to his second most preferred candidate; all candidates whose score exceeds the total number of voters face off in a runoff election. Leximin is a simple refinement of Nanson where a tie in the runoff is broken using the scores from the first round. In the generic case (which applies when the number of voters is large), maximin, Nanson, and leximin all coincide.

\begin{figure}[t]
	\centering
	\input{plot-irresolute}
	\caption{Fraction of anonymous profiles in which SCFs return more than one winner, computed using Ehrhart theory \citep[see, e.g.,][]{BGS15a} and the Normaliz package \citep{BIR+15a}.}
	\label{fig:irresolute-ehrhart}
\end{figure}
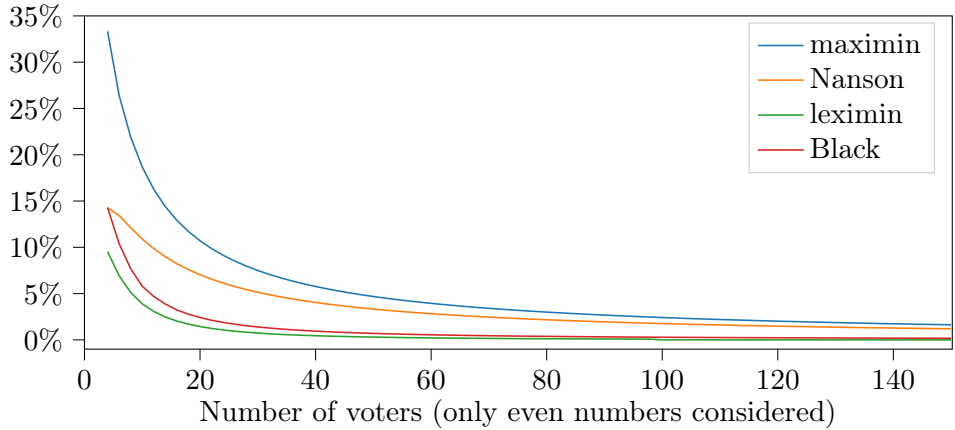

Studies on the frequency of voting paradoxes complement our results by showing that maximin (and its refinements) not only do well for small but also for large numbers of voters when there are three candidates. %
\citet{CMM14a} analyze the frequency of the reinforcement paradox of various Condorcet extensions using Monte Carlo simulations and find that ``although all frequencies are small, they are smaller for [maximin].''
\citet{Heil20a} proves that when the number of voters goes to infinity, maximin only suffers from the reinforcement paradox for 0.37\% of all pairs of anonymous profiles in which the winners coincide, which is lower than the corresponding numbers for Black's rule and plurality with runoff.
\citet{PlTi14a} analyze the frequency of voting paradoxes based on data generated using a spatial model, which they argue most accurately describes real-world preference profiles for three candidates. They conclude that ``the Black rule and the Nanson rule encounter most paradoxes and ties less frequently than the other rules do, especially in elections with few voters.'' As we have discussed, the leximin rule produces ties even less frequently than Nanson's rule. We can quantify this effect by computing the fraction of anonymous profiles on which different rules are non-resolute, as a function of the number of voters. The results are shown in \Cref{fig:irresolute-ehrhart} and show that leximin even outperforms Black's rule.%
\footnote{This can be explained by observing that for every social choice function $f$ that is pairwise, neutral, and never selects Pareto-dominated candidates, it holds that if leximin is irresolute on a profile $P$, then $f$ is also irresolute on $P$ (see Graphs A, B, and F in \Cref{tbl:rules-ordinal}). So leximin is resolute whenever possible.}

In conclusion, we believe that maximin and its refinements are very attractive for three-candidate elections and are compelling choices for adoption in real-world applications.\footnote{Early attempts to implement Nanson's rule in real-world elections were met with limited success \citep[see][for a historical account]{McLe96a}.}
Similar conclusions were drawn by \citet{Nurm89a}, \citet{FeNu18a}, and \citet{LeSm19a}.
We emphasize that our arguments do not extend beyond three-candidate elections. When there are four or more candidates, the no-show paradox cannot be avoided. In addition, there are many different ways of extending the three-candidate maximin rule to handle additional candidates, and not all of them are equally desirable. For example, for four or more candidates, maximin may return candidates that are last-ranked by a majority of voters \citep[see, e.g.,][]{Fels10a,BMS20b}, while other generalizations like split cycle avoid this problem.

\subsection*{Acknowledgments}
	This material is based on work supported by the Deutsche Forschungsgemeinschaft under grants {BR~2312/11-2} and {BR~2312/12-1} and by the Agence Nationale de la Recherche under grant ANR-22-CE26-0019 (CITIZENS). The authors are grateful to former students Florian Grundbacher, Keyvan Kardel, and Christian Stricker for constructing some of the examples used in this paper.

\end{document}

%% file: rule-table-ordinal-margins.tex
\tikzset{
    ordinal margin graph vertex/.style={circle,draw,font=\large,minimum size=16pt,inner sep=0pt},
    small margin/.style={line width=0.4pt, -{Latex[open,length=5pt]}},
    medium margin/.style={line width=1.13pt, -{Latex[length=5pt]}},
    large margin/.style={line width=2pt, -{Latex[length=6pt]}}
}
\begin{tabular}{lcccccccccccc}
\toprule
& Graph A & Graph B & Graph C & Graph D & Graph E & Graph F & Graph G & Graph H & Graph I & Graph J & Graph K & Graph L \\[2pt]
& \href{https://voting.ml/?profile=1ABC-1BCA-1CAB}{
\begin{tikzpicture}[scale=0.8, transform shape]
    \node[ordinal margin graph vertex] (v0) at (0, 0) {$a$};
    \node[ordinal margin graph vertex] (v1) at (1.5, 0) {$b$};
    \node[ordinal margin graph vertex] (v2) at (0.75, 1.3) {$c$};
    \draw[small margin] (v0) -- (v1);
    \draw[small margin] (v1) -- (v2);
    \draw[small margin] (v2) -- (v0);
\end{tikzpicture}}
& \href{https://voting.ml/?profile=1ABC-1ACB-1BAC-1BCA-1CAB-1CBA}{
\begin{tikzpicture}[scale=0.8, transform shape]
    \node[ordinal margin graph vertex] (v0) at (0, 0) {$a$};
    \node[ordinal margin graph vertex] (v1) at (1.5, 0) {$b$};
    \node[ordinal margin graph vertex] (v2) at (0.75, 1.3) {$c$};
\end{tikzpicture}}
& \href{https://voting.ml/?profile=1ABC-1ABC-1BCA-1CAB-1CAB}{
\begin{tikzpicture}[scale=0.8, transform shape]
    \node[ordinal margin graph vertex] (v0) at (0, 0) {$a$};
    \node[ordinal margin graph vertex] (v1) at (1.5, 0) {$b$};
    \node[ordinal margin graph vertex] (v2) at (0.75, 1.3) {$c$};
    \draw[medium margin] (v0) -- (v1);
    \draw[small margin] (v1) -- (v2);
    \draw[small margin] (v2) -- (v0);
\end{tikzpicture}}
& \href{https://voting.ml/?profile=1ABC-1CAB}{
\begin{tikzpicture}[scale=0.8, transform shape]
    \node[ordinal margin graph vertex] (v0) at (0, 0) {$a$};
    \node[ordinal margin graph vertex] (v1) at (1.5, 0) {$b$};
    \node[ordinal margin graph vertex] (v2) at (0.75, 1.3) {$c$};
    \draw[small margin] (v0) -- (v1);
\end{tikzpicture}}
& \href{https://voting.ml/?profile=1ACB-1CAB}{
\begin{tikzpicture}[scale=0.8, transform shape]
    \node[ordinal margin graph vertex] (v0) at (0, 0) {$a$};
    \node[ordinal margin graph vertex] (v1) at (1.5, 0) {$b$};
    \node[ordinal margin graph vertex] (v2) at (0.75, 1.3) {$c$};
    \draw[small margin] (v0) -- (v1);
    \draw[small margin] (v2) -- (v1);
\end{tikzpicture}}
& \href{https://voting.ml/?profile=1ABC-1ACB-1CAB-1CAB}{
\begin{tikzpicture}[scale=0.8, transform shape]
    \node[ordinal margin graph vertex] (v0) at (0, 0) {$a$};
    \node[ordinal margin graph vertex] (v1) at (1.5, 0) {$b$};
    \node[ordinal margin graph vertex] (v2) at (0.75, 1.3) {$c$};
    \draw[medium margin] (v0) -- (v1);
    \draw[small margin] (v2) -- (v1);
\end{tikzpicture}}
& \href{https://voting.ml/?profile=1ABC-1ABC-1ABC-1ABC-1BCA-1BCA-1CAB-1CAB-1CAB}{
\begin{tikzpicture}[scale=0.8, transform shape]
    \node[ordinal margin graph vertex] (v0) at (0, 0) {$a$};
    \node[ordinal margin graph vertex] (v1) at (1.5, 0) {$b$};
    \node[ordinal margin graph vertex] (v2) at (0.75, 1.3) {$c$};
    \draw[large margin] (v0) -- (v1);
    \draw[medium margin] (v1) -- (v2);
    \draw[small margin] (v2) -- (v0);
\end{tikzpicture}}
& \href{https://voting.ml/?profile=1ABC-1ABC-1ABC-1BCA-1CAB-1CAB}{
\begin{tikzpicture}[scale=0.8, transform shape]
    \node[ordinal margin graph vertex] (v0) at (0, 0) {$a$};
    \node[ordinal margin graph vertex] (v1) at (1.5, 0) {$b$};
    \node[ordinal margin graph vertex] (v2) at (0.75, 1.3) {$c$};
    \draw[medium margin] (v0) -- (v1);
    \draw[small margin] (v1) -- (v2);
\end{tikzpicture}}
& \href{https://voting.ml/?profile=1ABC-1ABC-1ABC-1BCA-1BCA-1CAB-1CAB}{
\begin{tikzpicture}[scale=0.8, transform shape]
    \node[ordinal margin graph vertex] (v0) at (0, 0) {$a$};
    \node[ordinal margin graph vertex] (v1) at (1.5, 0) {$b$};
    \node[ordinal margin graph vertex] (v2) at (0.75, 1.3) {$c$};
    \draw[medium margin] (v0) -- (v1);
    \draw[medium margin] (v1) -- (v2);
    \draw[small margin] (v2) -- (v0);
\end{tikzpicture}}
& \href{https://voting.ml/?profile=1ABC-1ABC-1BCA-1CAB}{
\begin{tikzpicture}[scale=0.8, transform shape]
    \node[ordinal margin graph vertex] (v0) at (0, 0) {$a$};
    \node[ordinal margin graph vertex] (v1) at (1.5, 0) {$b$};
    \node[ordinal margin graph vertex] (v2) at (0.75, 1.3) {$c$};
    \draw[small margin] (v0) -- (v1);
    \draw[small margin] (v1) -- (v2);
\end{tikzpicture}}
& \href{https://voting.ml/?profile=1ABC-1ABC-1ABC-1BCA-1BCA-1CAB-1CAB-1CAB-1CAB}{
\begin{tikzpicture}[scale=0.8, transform shape]
    \node[ordinal margin graph vertex] (v0) at (0, 0) {$a$};
    \node[ordinal margin graph vertex] (v1) at (1.5, 0) {$b$};
    \node[ordinal margin graph vertex] (v2) at (0.75, 1.3) {$c$};
    \draw[large margin] (v1) -- (v2);
    \draw[medium margin] (v0) -- (v1);
    \draw[small margin] (v2) -- (v0);
\end{tikzpicture}}
& \href{https://voting.ml/?profile=1ABC-1ABC-1BCA-1CAB-1CAB-1CAB}{
\begin{tikzpicture}[scale=0.8, transform shape]
    \node[ordinal margin graph vertex] (v0) at (0, 0) {$a$};
    \node[ordinal margin graph vertex] (v1) at (1.5, 0) {$b$};
    \node[ordinal margin graph vertex] (v2) at (0.75, 1.3) {$c$};
    \draw[medium margin] (v1) -- (v2);
    \draw[small margin] (v0) -- (v1);
\end{tikzpicture}}
\\
& $\begin{aligned} &\: m_{a,b} \\ = &\: m_{b,c} \\ = &\: m_{c,a} \\ > &\: 0 \end{aligned}$
& $\begin{aligned} &\: m_{a,b} \\ = &\: m_{b,c} \\ = &\: m_{c,a} \\ = &\: 0 \end{aligned}$
& $\begin{aligned} &\: m_{a,b} \\ > &\: m_{b,c} \\ = &\: m_{c,a} \\ > &\: 0 \end{aligned}$ 
& $\begin{aligned} &\: m_{a,b} \\ > &\: m_{b,c} \\ = &\: m_{c,a} \\ = &\: 0 \end{aligned}$ 
& $\begin{aligned} &\: m_{a,b} \\ = &\: m_{c,b} \\ > &\: m_{c,a} \\ = &\: 0 \end{aligned}$
& $\begin{aligned} &\: m_{a,b} \\ > &\: m_{c,b} \\ > &\: m_{c,a} \\ = &\: 0 \end{aligned}$ 
& $\begin{aligned} &\: m_{a,b} \\ > &\: m_{b,c} \\ > &\: m_{c,a} \\ > &\: 0 \end{aligned}$ 
& $\begin{aligned} &\: m_{a,b} \\ > &\: m_{b,c} \\ > &\: m_{c,a} \\ = &\: 0 \end{aligned}$ 
& $\begin{aligned} &\: m_{a,b} \\ = &\: m_{b,c} \\ > &\: m_{c,a} \\ > &\: 0 \end{aligned}$ 
& $\begin{aligned} &\: m_{a,b} \\ = &\: m_{b,c} \\ > &\: m_{c,a} \\ = &\: 0 \end{aligned}$ 
& $\begin{aligned} &\: m_{b,c} \\ > &\: m_{a,b} \\ > &\: m_{c,a} \\ > &\: 0 \end{aligned}$ 
& $\begin{aligned} &\: m_{b,c} \\ > &\: m_{a,b} \\ > &\: m_{c,a} \\ = &\: 0 \end{aligned}$ 
\\
\midrule
\makecell{top cycle}
&
$\{a,b,c\}$&$\{a,b,c\}$&$\{a,b,c\}$&$\{a,b,c\}$&$\{a,c\}$&$\{a,c\}$&$\{a,b,c\}$&$\{a,b,c\}$&$\{a,b,c\}$&$\{a,b,c\}$&$\{a,b,c\}$&$\{a,b,c\}$
\\
\midrule
\makecell{UC McKelvey}
&
$\{a,b,c\}$&$\{a,b,c\}$&$\{a,b,c\}$&$\{a,c\}$&$\{a,c\}$&$\{a,c\}$&$\{a,b,c\}$&$\{a,b,c\}$&$\{a,b,c\}$&$\{a,b,c\}$&$\{a,b,c\}$&$\{a,b,c\}$
\\
\midrule
\makecell{UC Bordes\\Banks}
&
$\{a,b,c\}$&$\{a,b,c\}$&$\{a,b,c\}$&$\{a,c\}$&$\{a,c\}$&$\{a,c\}$&$\{a,b,c\}$&$\{a,b\}$&$\{a,b,c\}$&$\{a,b\}$&$\{a,b,c\}$&$\{a,b\}$
\\
\midrule
\makecell{UC Gillies}
&
$\{a,b,c\}$&$\{a,b,c\}$&$\{a,b,c\}$&$\{a,c\}$&$\{a,c\}$&$\{a,c\}$&$\{a,b,c\}$&$\{a,c\}$&$\{a,b,c\}$&$\{a,c\}$&$\{a,b,c\}$&$\{a,c\}$
\\
\midrule
\makecell{defensible set}
&
$\{a,b,c\}$&$\{a,b,c\}$&$\{a,c\}$&$\{a,c\}$&$\{a,c\}$&$\{a,c\}$&$\{a,c\}$&$\{a,c\}$&$\{a,c\}$&$\{a,c\}$&$\{a\}$&$\{a\}$
\\
\midrule
\makecell{Llull\\Schwartz\\UC Fishburn}
&
$\{a,b,c\}$&$\{a,b,c\}$&$\{a,b,c\}$&$\{a,c\}$&$\{a,c\}$&$\{a,c\}$&$\{a,b,c\}$&$\{a\}$&$\{a,b,c\}$&$\{a\}$&$\{a,b,c\}$&$\{a\}$
\\
\midrule
\makecell{Copeland}
&
$\{a,b,c\}$&$\{a,b,c\}$&$\{a,b,c\}$&$\{a\}$&$\{a,c\}$&$\{a,c\}$&$\{a,b,c\}$&$\{a\}$&$\{a,b,c\}$&$\{a\}$&$\{a,b,c\}$&$\{a\}$
\\
\midrule
\makecell{maximin\\ranked pairs\\beat path\\split cycle}
&
$\{a,b,c\}$&$\{a,b,c\}$&$\{a,c\}$&$\{a,c\}$&$\{a,c\}$&$\{a,c\}$&$\{a\}$&$\{a\}$&$\{a\}$&$\{a\}$&$\{a\}$&$\{a\}$
\\
\midrule
\makecell{strict Nanson}
&
$\{a,b,c\}$&$\{a,b,c\}$&$\{c\}$&$\{a,c\}$&$\{a,c\}$&$\{a,c\}$&$\{a\}$&$\{a\}$&$\{a\}$&$\{a\}$&$\{a\}$&$\{a\}$
\\
\midrule
\makecell{stable voting}
&
$\{a,b,c\}$&$\{a,b,c\}$&$\{a,c\}$&$\{a\}$&$\{a,c\}$&$\{a\}$&$\{a\}$&$\{a\}$&$\{a\}$&$\{a\}$&$\{a\}$&$\{a\}$
\\
\midrule
\makecell{Nanson}
&
$\{a,b,c\}$&$\{a,b,c\}$&$\{a\}$&$\{a\}$&$\{a,c\}$&$\{a,c\}$&$\{a\}$&$\{a\}$&$\{a\}$&$\{a\}$&$\{a\}$&$\{a\}$
\\
\midrule
\makecell{leximin}
&
$\{a,b,c\}$&$\{a,b,c\}$&$\{a\}$&$\{a\}$&$\{a,c\}$&$\{a\}$&$\{a\}$&$\{a\}$&$\{a\}$&$\{a\}$&$\{a\}$&$\{a\}$
\\
\bottomrule
\end{tabular}

%% file: plot-irresolute.tex
\begin{tikzpicture}

\definecolor{crimson2143940}{RGB}{214,39,40}
\definecolor{darkgray176}{RGB}{176,176,176}
\definecolor{darkorange25512714}{RGB}{255,127,14}
\definecolor{forestgreen4416044}{RGB}{44,160,44}
\definecolor{lightgray204}{RGB}{204,204,204}
\definecolor{steelblue31119180}{RGB}{31,119,180}

\begin{axis}[
legend cell align={left},
legend style={fill opacity=0.8, draw opacity=1, text opacity=1, draw=lightgray204},
tick align=outside,
tick pos=left,
x grid style={darkgray176},
xlabel={Number of voters (only even numbers considered)},
xmin=0, xmax=150.2,
xtick style={color=black},
xtick distance=20,
y grid style={darkgray176},
ymin=-0.01, ymax=0.35,
ytick style={color=black},
ytick={-0.05,0,0.05,0.1,0.15,0.2,0.25,0.3,0.35},
yticklabels={-5\%,0\%,5\%,10\%,15\%,20\%,25\%,30\%,35\%},
height=6cm,
width=13cm
]
\addplot [semithick, steelblue31119180]
table {%
4 0.333333
6 0.264069
8 0.219114
10 0.18681
12 0.16257
14 0.14396
16 0.12914
18 0.11706
20 0.10705
22 0.09862
24 0.09141
26 0.08519
28 0.07976
30 0.07497
32 0.07073
34 0.06695
36 0.06354
38 0.06047
40 0.05768
42 0.05513
44 0.05281
46 0.05067
48 0.04869
50 0.04687
52 0.04517
54 0.04360
56 0.04213
58 0.04075
60 0.03947
62 0.03826
64 0.03712
66 0.03605
68 0.03504
70 0.03408
72 0.03318
74 0.03232
76 0.03151
78 0.03073
80 0.02999
82 0.02929
84 0.02862
86 0.02798
88 0.02737
90 0.02678
92 0.02622
94 0.02568
96 0.02516
98 0.02467
100 0.0241
102 0.0237
104 0.0232
106 0.0228
108 0.0224
110 0.0220
112 0.0216
114 0.0213
116 0.0209
118 0.0206
120 0.0202
122 0.0199
124 0.0196
126 0.0193
128 0.0190
130 0.0187
132 0.0184
134 0.0182
136 0.0179
138 0.0176
140 0.0174
142 0.0172
144 0.0169
146 0.0167
148 0.0165
150 0.0163
};
\addlegendentry{maximin}
\addplot [semithick, darkorange25512714]
table {%
4 0.142857
6 0.134199
8 0.121212
10 0.10889
12 0.09857
14 0.08978
16 0.08226
18 0.07587
20 0.07035
22 0.06555
24 0.06135
26 0.05764
28 0.05435
30 0.05141
32 0.04876
34 0.04637
36 0.04421
38 0.04223
40 0.04042
42 0.03876
44 0.03723
46 0.03582
48 0.03450
50 0.03328
52 0.03215
54 0.03108
56 0.03009
58 0.02916
60 0.02828
62 0.02745
64 0.02667
66 0.02594
68 0.02524
70 0.02458
72 0.02395
74 0.02336
76 0.02279
78 0.02225
80 0.02174
82 0.02124
84 0.02077
86 0.02032
88 0.01989
90 0.01948
92 0.01908
94 0.01870
96 0.01834
98 0.01799
100 0.0176
102 0.0173
104 0.0170
106 0.0167
108 0.0164
110 0.0161
112 0.0158
114 0.0155
116 0.0153
118 0.0150
120 0.0148
122 0.0146
124 0.0144
126 0.0141
128 0.0139
130 0.0137
132 0.0135
134 0.0133
136 0.0131
138 0.0130
140 0.0128
142 0.0126
144 0.0124
146 0.0123
148 0.0121
150 0.0120
};
\addlegendentry{Nanson}
\addplot [semithick, forestgreen4416044]
table {%
4 0.09523
6 0.06926
8 0.05128
10 0.0389
12 0.0307
14 0.0247
16 0.0203
18 0.0170
20 0.0144
22 0.0124
24 0.0107
26 0.0094
28 0.0083
30 0.0074
32 0.0066
34 0.0059
36 0.0054
38 0.0049
40 0.0045
42 0.0041
44 0.0038
46 0.0035
48 0.0032
50 0.0030
52 0.0028
54 0.0026
56 0.0024
58 0.0023
60 0.0021
62 0.0020
64 0.0019
66 0.0018
68 0.0017
70 0.0016
72 0.0015
74 0.0014
76 0.0013
78 0.0013
80 0.0012
82 0.0012
84 0.0011
86 0.0011
88 0.0010
90 0.0010
92 0.0009
94 0.0009
96 0.0008
98 0.0008
100 0.000
102 0.000
104 0.000
106 0.000
108 0.000
110 0.000
112 0.000
114 0.000
116 0.000
118 0.000
120 0.000
122 0.000
124 0.000
126 0.000
128 0.000
130 0.000
132 0.000
134 0.000
136 0.000
138 0.000
140 0.000
142 0.000
144 0.000
146 0.000
148 0.000
150 0.000
};
\addlegendentry{leximin}
\addplot [semithick, crimson2143940]
table {%
4 0.14285
6 0.10389
8 0.07692
10 0.05794
12 0.04686
14 0.03869
16 0.03228
18 0.02775
20 0.02416
22 0.02118
24 0.01887
26 0.01695
28 0.01529
30 0.01394
32 0.01278
34 0.01175
36 0.01088
38 0.01012
40 0.00943
42 0.00883
44 0.00830
46 0.00782
48 0.00739
50 0.00700
52 0.00664
54 0.00631
56 0.00602
58 0.00574
60 0.00549
62 0.00526
64 0.00505
66 0.00485
68 0.00467
70 0.00449
72 0.00433
74 0.00418
76 0.00404
78 0.00391
80 0.00378
82 0.00367
84 0.00356
86 0.00345
88 0.00335
90 0.00326
92 0.00317
94 0.00308
96 0.00300
98 0.00293
100 0.002859
102 0.002789
104 0.002723
106 0.002659
108 0.002599
110 0.002541
112 0.002485
114 0.002431
116 0.002380
118 0.002331
120 0.002283
122 0.002238
124 0.002194
126 0.002151
128 0.002111
130 0.002071
132 0.002034
134 0.001997
136 0.001962
138 0.001928
140 0.001894
142 0.001862
144 0.001832
146 0.001802
148 0.001773
150 0.001744
};
\addlegendentry{Black}
\end{axis}

\end{tikzpicture}